\newif\iffinal
\finaltrue

\documentclass[acmsmall,nonacm]{acmart}
\usepackage[T1]{fontenc}
\usepackage[utf8]{inputenc}
\usepackage[english]{babel}

\def\RATIOP{205} \def\RATIOG{34}  \def\RATIO{616}   

\def\MetaI{\hyperref[th:metatheorem_nonpoly]{Meta-Theorem~I}\xspace}
\def\MetaII{\hyperref[th:metatheorem_nonpaddable]{Meta-Theorem~II}\xspace}
\def\MetaIII{\hyperref[th:metatheorem_paddable]{Meta-Theorem~III}\xspace}

\usepackage{graphicx,framed,calc,amsmath,amsthm,algorithm,algpseudocode}
\usepackage{wrapfig,etoolbox,anyfontsize,bm,textcase}
\usepackage{thmtools}
\usepackage{enumitem}
\usepackage{balance}
\usepackage[dvipsnames]{xcolor}
\usepackage{float,csquotes,xspace}

\usepackage{multicol,multirow,varwidth,outlines}

\usepackage{tikz}
\usetikzlibrary{math,patterns,positioning,arrows,arrows.meta,graphs,graphs.standard,decorations.pathmorphing,tikzmark,calc,hobby}
\tikzset{
    vertex/.style={circle, fill=white, line width=2pt, draw=tertiary}
}
\tikzset{
    fvertex/.style={circle, fill=primary!50!black, line width=2pt, draw=white}
}
\tikzset{
    edge/.style={very thick, draw=primary!50!black}
}

\iffinal
\long\def\cyril#1{}
\long\def\tim#1{}
\long\def\marthe#1{}
\long\def\alex#1{}
\else
\long\def\cyril#1{\textcolor{blue}{[Cyril: {#1}]}}
\long\def\tim#1{\textcolor{teal}{[Timothé: {#1}]}}
\long\def\marthe#1{\todo[read]{Marthe: {#1}}}
\long\def\alex#1{\textcolor{red}{[Alexandra: {#1}]}}
\fi

\definecolor{primary}{HTML}{255077}
\colorlet{secondary}{Dandelion}
\definecolor{tertiary}{HTML}{fe4e63}

\colorlet{mygreen}{YellowGreen}
\definecolor{myblack}{HTML}{255077}
\colorlet{myorange}{Dandelion}
\colorlet{myblue}{CornflowerBlue}
\colorlet{mypink}{Lavender}
\colorlet{mydarkblue}{RoyalBlue}
\definecolor{myred}{HTML}{fe4e63}

\newcommand{\eps}{\varepsilon}
\def\emo{{\eps}^{-1}}
\newcommand{\LOCAL}{\textsf{LOCAL}\xspace}

\DeclareMathOperator{\MDS}{MDS}
\DeclareMathOperator{\diam}{diam}
\DeclareMathOperator{\asdim}{asdim}
\DeclareMathOperator{\poly}{poly}
\DeclareMathOperator{\OPT}{OPT}

\NewDocumentCommand\set{sm}{\IfBooleanTF#1{\{{#2}\}}{\left\{{#2}\right\}}}
\NewDocumentCommand\ceil{sm}{\IfBooleanTF#1{\lceil{#2}\rceil}{\left\lceil{#2}\right\rceil}}
\NewDocumentCommand\floor{sm}{\IfBooleanTF#1{\lfloor{#2}\rfloor}{\left\lfloor{#2}\right\rfloor}}
\NewDocumentCommand\pare{sm}{\IfBooleanTF#1{({#2})}{\left({#2}\right)}}
\NewDocumentCommand\range{smm}{\IfBooleanTF#1{\set*{{#2},\dots,{#3}}}{\set{{#2},\dots,{#3}}}}
\newcommand{\card}[1]{\left|{#1}\right|}

\def\poly{\mathrm{poly}}

\def\sfA{\mathsf{A}}
\def\sfB{\mathsf{B}}

\def\sC{\mathscr{C}}
\def\sD{\mathscr{D}}
\def\sG{\mathscr{G}}

\def\bbN{\mathbb{N}}
\def\minDS{\textsc{Minimum Dominating Set}\xspace}

\let\le\leqslant
\let\ge\geqslant
\let\leq\leqslant
\let\geq\geqslant

\let\emptyset\varnothing

\def\CC{\mathfrak{C}}

\makeatletter

\renewenvironment{thmt@restatable}[3][]{\thmt@toks{}

  \stepcounter{thmt@dummyctr}

  \long\def\thmrst@store##1{\expandafter\gdef\csname #3\endcsname{\@ifstar{\thmt@thisistheonefalse\csname thmt@stored@#3\endcsname
      }{\thmt@thisistheonetrue\csname thmt@stored@#3\endcsname
      }}\expandafter\long\expandafter\gdef
      \csname thmt@stored@#3\expandafter\endcsname\expandafter{\begingroup
      \ifthmt@thisistheone
\thmt@rst@storecounters{#3}\else

\@ifpackageloaded{hyperref}{\expandafter\protected@edef\csname the#2\endcsname{\noexpand\hyperref[thmt@@#3]{\thmt@trivialref{thmt@@#3}{??}}}}{\expandafter\protected@edef\csname the#2\endcsname{\thmt@trivialref{thmt@@#3}{??}}}

\ifcsname r@thmt@@#3\endcsname\else
          \G@refundefinedtrue
        \fi

\expandafter\let\csname c@#2\endcsname=\c@thmt@dummyctr
        \expandafter\let\csname theH#2\endcsname=\theHthmt@dummyctr

\let\label=\thmt@gobble@label
        \let\ltx@label=\@gobble 

\def\thmt@restorecounters{}\@for\thmt@ctr:=\thmt@innercounters\do{\protected@edef\thmt@restorecounters{\thmt@restorecounters
            \protect\setcounter{\thmt@ctr}{\arabic{\thmt@ctr}}}}

\thmt@trivialref{thmt@@#3@data}{}\fi

\ifthmt@restatethis
        \thmt@restatethisfalse
      \else
        \csname #2\expandafter\endcsname\ifx\@nx#1\@nx\else[{#1}]\fi
      \fi

\ifthmt@thisistheone
        \label{thmt@@#3}\fi

##1\csname end#2\endcsname

\ifthmt@thisistheone\else\thmt@restorecounters\fi
      \endgroup
    }

\csname #3\expandafter\endcsname\ifthmt@thisistheone\else*\fi

\expandafter\end\expandafter{\@currenvir}}

  \thmt@collect@body\thmrst@store
}{}

\makeatother 
\acmYear{2026}\copyrightyear{2026}
\setcopyright{acmlicensed}
\acmConference[PODC '26]{x}{y}


\def\myfootnotemark#1#2{{\let\thefootnote\relax\footnotemark\addtocounter{footnote}{-1}\hspace{-.5ex}}\textsuperscript{\hyperref[#1:#2]{\the\numexpr\thefootnote+#2}}}

\definecolor{ourresults}{RGB}{200,0,0}

\usepackage[nameinlink,capitalize]{cleveref}
\usepackage[norefs,nocites]{refcheck}
\makeatletter
\newcommand{\refcheckize}[1]{\expandafter\let\csname @@\string#1\endcsname#1\expandafter\DeclareRobustCommand\csname relax\string#1\endcsname[1]{\csname @@\string#1\endcsname{##1}\wrtusdrf{##1}}\expandafter\let\expandafter#1\csname relax\string#1\endcsname
}
\makeatother

\refcheckize{\cref}
\refcheckize{\Cref}

\newtheorem{theorem}{Theorem}[section]

\newtheorem{proposition}[theorem]{Proposition}

\newtheorem{claim}[theorem]{Claim}

\colorlet{grun}{mygreen}
\tikzstyle{snode}=[shape= circle, draw=myred, ultra thick, fill=myred!30]
\tikzstyle{special node} = [shape= circle, draw=blue, ultra thick, fill=blue!10]
\tikzstyle{blue node} = [special node]
\tikzstyle{yellow node} = [shape= circle, draw=gelb, ultra thick, fill=gelb!30]
\tikzstyle{red node} = [snode]
\tikzstyle{green node} =  [shape= circle, draw=mygreen, ultra thick, fill=mygreen!20]
\tikzstyle{maybe node} = [shape = circle, draw=black, dotted, thick]
\tikzstyle{maybe new node} =[shape = circle, draw=grunn, dotted, ultra thick, fill=blue!10]
\tikzstyle{maybe equal} = [dotted] \tikzstyle{normal node} =[shape=circle,draw=primary, text=primary]

\newenvironment{proofofclaim}[1][Proof of claim]{\begin{proof}[#1]}{\qedhere
  \end{proof}}

\let\le\leqslant
\let\ge\geqslant
\let\leq\leqslant
\let\geq\geqslant

\title{Meta-Theorems for Cuttable Distributed Problems}

\author{Marthe Bonamy}
\orcid{0000-0001-7905-8018}
\affiliation{\institution{LaBRI, University of Bordeaux, CNRS}
  \city{Bordeaux}
  \country{France}
}
\email{marthe.bonamy@u-bordeaux.fr}

\author{Avinandan Das}
\orcid{0000-0002-7471-7499}
\affiliation{\institution{Aalto University}
  \city{Espoo}
  \country{Finland}
}
\email{avinandan.das@aalto.fi}

\author{Cyril Gavoille}
\orcid{0000-0003-3671-8607}
\affiliation{\institution{LaBRI, University of Bordeaux, CNRS}
  \city{Bordeaux}
  \country{France}
}
\email{gavoille@labri.fr}

\author{Timothé Picavet}
\orcid{0000-0002-7129-0127}
\affiliation{\institution{LaBRI, University of Bordeaux, CNRS}
  \city{Bordeaux}
  \country{France}
}
\email{timothe.picavet@u-bordeaux.fr}

\author{Jukka Suomela}
\orcid{0000-0001-6117-8089} 
\affiliation{\institution{Aalto University}
  \city{Espoo}
  \country{Finland}
}
\email{jukka.suomela@aalto.fi}

\author{Alexandra Wesolek}
\orcid{0000-0003-4841-5937}
\affiliation{\institution{LaBRI, University of Bordeaux, CNRS}
  \city{Bordeaux}
  \country{France}
}
\email{alexandra.wesolek@u-bordeaux.fr}

\date{}
\authorsaddresses{}

\def\minDS{\textsc{Minimum Dominating Set}\xspace}
\def\minCDS{\textsc{Minimum Connected Dominating Set}\xspace}

\def\minkTDS{\textsc{Minimum $k$-Tuple Dominating Set}\xspace}

\begin{abstract}
     We prove that given any $\alpha$-approximation \LOCAL algorithm for \minDS (MDS) on planar graphs, we can construct an $f(g)$-round $(3\alpha+1)$-approximation \LOCAL algorithm for MDS on graphs embeddable in a given Euler genus-$g$ surface. Heydt et al. [\textit{European Journal of Combinatorics} (2025)] gave an algorithm with $\alpha=11+\eps$, from which we derive a $(\RATIOG +\eps)$-approximation algorithm for graphs of genus $g$, therefore improving upon the current state of the art of~$24g+O(1)$ due to Amiri et al. [\textit{ACM Transactions on Algorithms} (2019)]. It also improves the approximation ratio of~$91+\eps$ due to Czygrinow et al. [\textit{Theoretical Computer Science} (2019)] in the particular case of orientable surfaces.

     We generalize this result into two directions: (1) by considering other graph problems studied in Distributed Computing such as \minkTDS, for which constant-round approximation algorithms were known for planar graphs, but not for graphs of bounded genus; and (2) by considering graph classes beyond bounded genus graphs, called \emph{locally nice}, and relying on the asymptotic dimension of the class. We prove these results by a series of meta-theorems about \emph{cuttable} minimization problems with constant-round approximation \LOCAL algorithms. Roughly speaking, in cuttable problems, one can systematically extract small subgraphs whose solutions are in proportion to the global solution restricted to the neighbourhood of the subgraph.
\end{abstract}

\begin{document}


\maketitle

\section{Introduction}
\label{sec:intro}

\minDS (MDS) is a famous minimization problem on graphs, known to be NP-complete even in cubic planar graphs~\cite{GJ79,KYK80}. The goal is to find a smallest subset of vertices of the input graph that intersects all radius-1 balls of the graph.

In this paper, we focus our attention on distributed algorithms that can approximate MDS on the graph underlying the topology of the network. Due to the covering property of a dominating set, the problem and its variants (like \textsc{Connected Dominating Set}~\cite{WAF02}) get increasingly more attention in Networking, in particular for mobile and ad-hoc networks. Not only are mobile networks important for point-to-point communications, but specialized ad-hoc networks, such as sensor networks, are important for environmental monitoring tasks. We refer to \cite{KW05,KWZ09} for extended discussions about the motivations of MDS for routing purposes of such networks.

In the \LOCAL model, approximating MDS up to a constant factor in general $n$-vertex graphs is known to require $\Omega(\sqrt{\log{n}/\log\log{n}}\,)$ rounds~\cite{KMW16,CL21b}. 
On the positive side, it is possible to $(1+\eps)$-approximate MDS for any graph in $\poly(\emo\log{n})$ rounds by combining the techniques of~\cite{GKM17} and of~\cite[Corollary~3.11]{RG20}.

For specific graphs, better round complexities can be achieved; cf.~\cite{LPW13,Suomela13} for a large collection of results (including unit-disk graphs, and planar graphs as a basis of the Gabriel graph model widely used in ad-hoc networks~\cite{WY07}). For instance, $O(\log^*{n})$ rounds suffice for planar graphs~\cite{CHW08a}, or more generally for $K_t$-minor-free graphs~\cite{CHW18} and for sub-logarithmic expansion graphs\footnote{Roughly speaking, the expansion of a graph $G$ is a function $f$ that, for any $r$, bounds the maximum edge density of a depth-$r$ minor of $G$ by $f(r)$, where a depth-$r$ minor is a minor of $G$ where each branch set has radius at most $r$.}~\cite{ASS19}, and $o(\log^*{n})$ rounds are not sufficient to get a $(1+\eps)$-approximation of MDS on cycles~\cite{CHW08a} or unit-disk graphs~\cite{LW08}. So, achieving constant-round algorithms must be at the price of relaxing the $(1+\eps)$ approximation ratio.

For planar graphs, the quest for constant-approximation and constant-round algorithms seems to start with \cite{LOW08}\cite[Chp.~13]{Lenzen11}\cite{LPW13}, with a $130$-approximation. Since then, a long line of research has been aimed at improving this ratio. In short, the best approximation ratio to date is $11 + \eps$, due to~\cite{HKOSV25}. We refer to the nice survey of~\cite{HKOSV25} for planar graphs and sub-families, including lower bounds.

In the meantime, the quest for constant-approximation ratio and constant-round distributed algorithms has been proposed for larger classes of graphs. However, very few examples are known of graph classes $\sC(p)$, depending on some fixed parameter $p$, where MDS can be solved by a \LOCAL algorithm with round complexity $f(p)$ and truly-constant approximation ratio (independent of $p$). 

Larger graph classes that make good candidates for extending planar graphs are $H$-minor-free graphs, with some graph $H$ not limited to $K_5$ or $K_{3,3}$. For $K_p$-minor-free graphs, we only know of an exponential approximation ratio in $p$~\cite{KSV21,HKOSV25}, whereas a linear approximation ratio could be possible. For $K_{3,p}$-minor-free graphs~\cite{HKOSV25}, the dependency is much better (linear in~$p$), but still not truly constant. Very recently, it was shown~\cite{BGPW25} that, for $K_{2,p}$-minor-free graphs, the approximation ratio is~$50$ regardless of the value of $p$. 

\def\rmM{\mathrm{M}}

On the other end of the spectrum, the narrowest way to meaningfully extend planar graphs is through embeddable graphs\footnote{I.e., that can be drawn on the surface without edge crossing, like planar graphs for the sphere.} on surfaces\footnote{That are compact, connected $2$-dimensional manifold without boundary.} of Euler genus $g$. Such surfaces can be obtained from a sphere by adding $g/2 \ge 0$ handles\footnote{By adding a handle, we mean that two disjoint disks of the sphere are replaced by a cylinder.} (if they are orientable) or by adding $g\ge 1$ cross-caps\footnote{By adding a cross-cap, we mean that a disk of the sphere is replaced by a M{\"o}bius strip.} (if they are nonorientable). The \emph{Euler genus} of a graph $G$ is the minimum number $g$ such that $G$ embeds on a surface of Euler genus~$g$. In this context, \cite[Theorem~3.11]{ASS19} proposed an $(24g+O(1))$-approximation algorithm with constant-round complexity (actually linear in $g$). For graphs of orientable genus $g$, i.e., embeddable on an orientable surface of genus $g$, \cite{CHWW19} designed a constant-round algorithm that returns a dominating set of size at most $91\rmM + 76g - 66$, where $\rmM$ is the optimal size. By adding a simple extra brute-force step in their algorithm, we can convert it into a $(91+\eps)$-approximation (cf. \cref{prop:aMDS+b}), so with a ratio independent of $g$.
However, the result of \cite{CHWW19} does not transfer to graphs of bounded Euler genus, as for any $g$, there are graphs embeddable on the projective plane\footnote{\label{foot:g=1}For instance, a cycle with $n = 2g+6$ vertices where opposite vertices are connected by an extra edge as Euler genus~1 but orientable genus at least $g+1 = n/2 - 2$, cf. \cite{ABY63}.} (so of Euler genus~1) that cannot be embedded on any orientable surface of genus~$g$ (so of orientable genus at least $g+1$).
So, for these graphs, the approximation ratio will be $91+\eps$, but with the number of rounds possibly depending on the number of vertices\textsuperscript{\ref{foot:g=1}}.

Graphs of Euler genus $g$ are included in $K_{3,2g+3}$-minor-free graphs, because the Euler genus of $K_{3,2g+3}$ is at least $g+1$ (cf. \cite[Theorem~4.4.7]{MT01}) and the class of Euler genus-$g$ graphs is closed under taking minors. We observe that the approximation ratio for $K_{3,p}$-minor-free graphs, for $p = 2g+3$, due to \cite[Theorem~2.2]{HKOSV25} provides an approximation ratio of $O(\sqrt{g})$. Indeed, the ratio of their algorithm is precisely $(2 + \eps)(2\nabla_1 + 1)$, where $\nabla_1$ is the maximal edge density of a depth-$1$ minor of any graph of the class. As we can check (cf. \cref{prop:nabla1})
that $\nabla_1 \le \sqrt{3g/2} + 3$, the approximation ratio is at most $4\sqrt{3g/2} + 14 + \eps$. As it may occur that $\nabla_1 \ge \sqrt{3g/2} - O(1)$ for some Euler genus-$g$ graph (cf. \cref{prop:nabla1}), the best known approximation ratio for MDS in Euler genus-$g$ graphs is $\Theta(\sqrt{g})$, and $91+\eps$ for orientable genus-$g$ graphs. 

It is important to note that no \LOCAL algorithms for Euler genus-$g$ graphs can achieve simultaneously constant round complexity and constant approximation ratio. Indeed, by taking $n = \Omega(\sqrt{g})$, since there exists $n$-vertex graphs of Euler genus $\Theta(n^2)$, the $\Omega(\sqrt{\log{n}/\log\log{n}}\,)$ lower bound of~\cite{KMW16,CL21b} yields a lower bound of $\Omega(\sqrt{\log{g}/\log\log{g}}\,)$. 

\paragraph{Our Contributions.}

In the remainder of the paper, by ``\LOCAL algorithm'', we mean deterministic distributed algorithm in the classical \LOCAL model.

\begin{itemize}
    \item We propose a $(34+\eps)$-approximation \LOCAL algorithm for MDS in Euler genus-$g$ graphs whose round complexity depends only on $g$ (cf. \cref{th:bounded_genus}). This improves approximation ratios in both the orientable (previously $91+\eps$) and non-orientable (previously $24g+O(1)$) genus case. The analysis of the algorithm uses the fact that such graphs have asymptotic dimension at most two.
    
    \item This result follows directly from more general meta-theorems. More generally, we prove that any constant-round $\alpha$-approximation for MDS on planar graphs can be turned into a $f(g)$-round $(3\alpha+1)$-approximation algorithm for MDS on Euler genus-$g$ graphs, for some function $f$.
    
    \item We extend this result in two ways. First, we show that our meta-theorems hold for all local minimization problems that are well-behaved (including MDS and its variants). We provide a simple graph-theoretic property on the problem that is sufficient for our meta-algorithms to apply. This significantly simplifies proofs and circumvents having to reason about various algorithms. 
    
    \item We apply our results on problems such as \minkTDS (cf. \cref{th:otherproblems}), obtaining constant-ratio approximations on problems where none were previously known. We also obtain a ratio of $O(\sqrt{g})$ for \minCDS in Euler genus-$g$ graphs.
    
    \item Second, we prove that our meta-theorems hold on much broader ``locally nice'' graph classes, such as locally bounded-genus graphs, or locally bounded-genus graphs with a bounded number of apex vertices (whose unconstrained neighborhoods can break the surface embeddability of the graph). 
    
    \item We prove three versions of our meta-theorem (see \MetaI, \MetaII and \MetaIII), making it easily applicable and adapted to different situations. 
\end{itemize}

Unlike similar notions, such as the classical network decomposition~\cite{AGLP89}, no preliminary construction is required for any of our algorithms. Obviously, they depend on the parameters of the graph classes on which they are applied, like the Euler genus $g$ and/or the asymptotic dimension of the class with its control function. 

\paragraph{Organization.} Basic notations and asymptotic dimension are introduced in \cref{sec:prelim}. \Cref{sec:bounded_genus} derives a $\RATIO$-approximation for MDS on bounded-genus graphs, gently introducing the ideas and motivations of the later meta-theorems. Finally, \cref{sec:metatheorems} presents our different meta-theorems, and their applications to different problems of the literature.

\section{Preliminaries}\label{sec:prelim}

\paragraph{Graphs and Domination.}

A subset $X\subseteq V(G)$ \emph{dominates} $S$ if $S\subseteq N[X]$. We denote by $\MDS(G,S)$ the minimum size of a subset of $V(G)$ dominating $S$. Note that $\MDS(G,V(G))$, denoted by $\MDS(G)$ for short, is the size of a minimum dominating set for $G$.

\paragraph{Asymptotic Dimension.}

In this section, we focus on defining asymptotic dimension\footnote{Asymptotic dimension was originally defined on metric spaces by Gromov~\cite{gromov1993geometric} in the field of geometric group theory.} (a non-trivial task, as it happens) and explaining how to exploit it, in the hope that others may be able to exploit it in turn.

A graph $G$ with a spanning subgraph $H$ is \emph{$k$-colorable $\delta$-bounded in $H$} if each vertex of $G$ can be assigned a color from $\range{0}{k-1}$ so that the distance in $H$ between any two vertices taken in a monochromatic path of $G$ is at most $\delta$. In other words, all monochromatic connected components of $G$ must have a weak diameter in $H$ at most $\delta$. 
A proper $k$-coloring of $G$ is nothing else than a $k$-coloring $0$-bounded in $H$. If $H$ has diameter $D$, then it is $1$-colorable $D$-bounded in $H$.

The asymptotic dimension of a graph class $\sC$ is a non-negative integer denoted by $\asdim(\sC)$, and related to the $k$-coloring boundedness of the $r$-th power of $G$. It is a notion closely related to \emph{sparse partitions}\footnote{However, those notions follow a different philosophy. In the study of asymptotic dimension, we want to minimize the dimension. In the study of sparse partitions and covers, the main goal is to minimize a function of the dimension and the diameter of the sets.} and \emph{weak sparse covers} (see~\cite[§1.8]{BBEGx24} for precise connection between these notions). Recall that the $r$-th power of $G$, denoted by $G^r$, is the graph obtained from $V(G)$ by adding edges between pairs of vertices at distance at most $r$ in $G$. 

Rather than giving its standard definition, we prefer the following equivalent one\footnote{In the original Proposition~1.17 of \cite{BBEGx24}, $G^r$ is $(d+1)$-colorable $f(r)$-bounded in $G^r$, instead of in $G$. By following this original definition, the function $f$ is not a control function, whereas in \cref{prop:asdim_def} it is.}:

\begin{proposition}[{\cite[Proposition~1.17]{BBEGx24}}]\label{prop:asdim_def}
For every graph class $\sC$, $\asdim(\sC) \le d$ if and only if there exists a function $f: \bbN \to \bbN$, called \emph{control function}, such that for every $G \in \sC$ and integer $r\ge 1$, $G^r$ is $(d+1)$-colorable $f(r)$-bounded in $G$. 
\end{proposition}

Any associated control function $f$ that makes $\asdim(\sC) = d$ is also called a $d$-dimensional control function of $\sC$. The Assouad-Nagata dimension of $\sC$ has a similar definition, except that the control function must be linear. A crude example is that paths have Assouad-Nagata dimension~$1$, because, for each $r\ge 1$, it suffices to alternatingly color subsets of $r$ consecutive vertices  with two colors, see \cref{fig:path-dim1}.

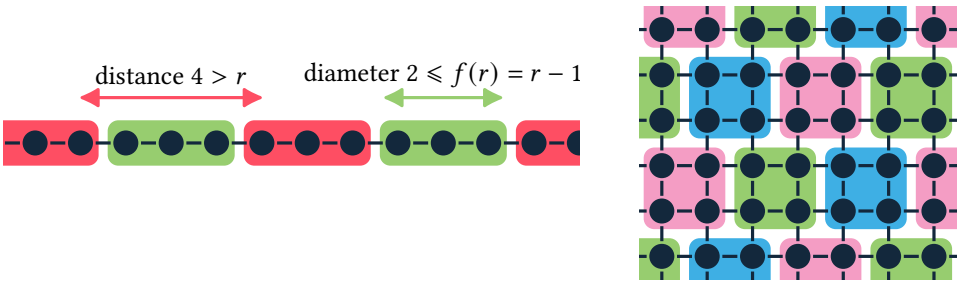
\begin{figure}[htpb!]
    \centering \begin{tikzpicture}[scale=0.6]
        \clip (3.3,-3) rectangle (16,2);
        \foreach \x in {0,...,4}
        {
            \draw[fill=mygreen,rounded corners,draw=none] (6*\x-0.4,-0.5) rectangle ++(2.8,1);
            \draw[fill=myred,rounded corners,draw=none] (6*\x+3-0.4,-0.5) rectangle ++(2.8,1);
        }

        \foreach \x in {0,...,16}
            \node[fvertex, draw=none] (\x) at (\x,0) {};
    
        \foreach \x [count=\xi] in {0,...,15}  
            \draw[edge] (\x)--(\xi);

        \draw[very thick,{Latex[round,length=3mm, width=3mm]}-{Latex[round,length=3mm, width=3mm]},draw=mygreen] (12-0.3,1) -- (14+0.3,1);
        \draw[very thick,{Latex[round,length=3mm, width=3mm]}-{Latex[round,length=3mm, width=3mm]},draw=myred] (5,1) --(9,1);
        \node[fill=none] at (13,1.5) {diameter $2 \le f(r) = r-1$};
        \node[fill=none] at (7,1.5) {distance $4 > r$}; 
    \end{tikzpicture} \hspace{2em}\begin{tikzpicture}[scale=0.6]
       \clip (0.5,0.5) rectangle (7.5,6.5);
       
       \foreach \x in {0,...,5}
       \foreach \y in {0,...,2}  
       \foreach \o in {0,1}
        {
            \draw[fill=mygreen,rounded corners,draw=none] (-3*\o+6*\x-0.4,2*\o+4*\y-0.4) rectangle ++(1.8,1.8);
            \draw[fill=myblue,rounded corners,draw=none] (-3*\o+6*\x+2-0.4,2*\o+4*\y-0.4) rectangle ++(1.8,1.8);
            \draw[fill=mypink,rounded corners,draw=none] (-3*\o+6*\x+4-0.4,2*\o+4*\y-0.4) rectangle ++(1.8,1.8);
            }
           
            \foreach \x in {0,...,15}
            \foreach \y in {0,...,7} 
            \node[fvertex,draw=none] (\x_\y) at (\x,\y) {};
           
            \foreach \x [count=\xi] in {0,...,14}
            \foreach \y [count=\yi] in {0,...,6}  
            \draw[edge] (\x_\y)--(\x_\yi) (\x_\y)--(\xi_\y) ;
        \end{tikzpicture}
    \caption{An example that shows that $P^r$ is $2$-colorable $(r-1)$-bounded in $P$ for $r = 3$ (on the left), and an example of a $3$-coloring of grids for $r = 2$ and $f(r) = 2(r-1) = 2$ (on the right).}\Description{~}
    \label{fig:path-dim1}
\end{figure}
More generally, $d$-dimensional grids have a $d$-dimensional control function $f(r) = d\cdot (r-1)$.

Trees, and more generally graph classes of bounded treewidth (resp. layered treewidth) have asymptotic dimension~$1$ (resp.~$2$). Planar graphs, and more generally the classes of $H$-minor-free graphs (for any fixed $H$) have asymptotic dimension~$2$ as shown by \cite{BBEGx24}: the dependency in $H$ only shows in the control function $f$. Dense graph classes may also have small asymptotic dimension. For example, it is sufficient that there is a quasi-isometry into a class having small asymptotic dimension\footnote{It is conjectured in \cite{BBEGx24} that, even more generally, any graph class forbidding some graph as a \emph{fat} minor should also have asymptotic dimension at most $2$.}. The dense class of chordal graphs, being quasi-isometric to the class of trees, has asymptotic dimension $1$.

Intuitively, each connected component of a color class in $G^r$ can be viewed as a ``simple'' component of the graph, that we already understand well.  
In the setting of \LOCAL algorithms with constant round complexity $r/2$, each vertex can observe only its distance-$(r/2)$ neighborhood.  
If the underlying graph class has asymptotic dimension $d$, this neighborhood intersects at most $d+1$ such simple parts.
As a result, we can derive strong approximation guarantees, since each vertex's view of the graph remains constrained and conceptually simple. We observe that all our algorithms make use of the control function $f$, but not for any particular pre-computing coloring or partitioning, so that the asymptotic dimension $d$ is actually used only in analysis of the approximation ratio.

\section{MDS on Bounded Euler Genus Graphs}\label{sec:bounded_genus}
This section explains how asymptotic dimension helps us devise approximation algorithms for MDS on Euler-genus-$g$ graphs. It also serves as an introduction to the method of cutting and motivates the later meta-theorems. 
The following algorithm for planar graphs will serve as a black-box building block for the more involved algorithm of \cref{th:bounded_genus_simple}.

\begin{restatable}{proposition}{ApproxPlanar}\label{prop:algo_planar}
    There is a $\RATIOP$-approximation \LOCAL algorithm $\sfA$ for \minDS in planar graphs with round complexity~$6$.
\end{restatable}

To present the algorithm (cf. \cref{algo:planar}) we first need to define the ``best'' minimum dominating sets of a labelled graph $H$. Among all dominating sets, we first discard those which contain some $v$ for which there exists $w$ with $N[v]\subsetneq N[w]$, and among the remaining ones, we declare ``best'' the one lexicographically smallest considering the labels of the vertices.

\begin{algorithm}
\caption{The algorithm for dominating set on planar graphs.}\label{algo:MDSplanar}
\begin{algorithmic}[1]

\Require A planar graph $G$.

\Ensure A dominating set $D$ of $G$ such that $\card{D}\le \RATIOP \cdot \MDS(G)$.

\State $D \gets \emptyset$

\State Each vertex $u$ computes $G[N^{5}[u]]$ and all minimum dominating sets of $N^4[u]$ in $G$. Among those sets, $u$ selects the ``best'' one as $D_u$.

\State Each vertex $u$ picks the vertex $v_u$ of smallest label in $D_u\cap N[u]$, and adds it to $D$. 

\end{algorithmic}
\label{algo:planar}
\end{algorithm}

By construction, Algorithm~\ref{algo:planar}  outputs a dominating set $D$ within~$6$ rounds. In fact, Algorithm~\ref{algo:planar}  is very natural in the sense that in order to compute a global dominating set we use local dominating sets. The proof on the approximation factor is given in the Appendix~\ref{appendix:planar}. 
We extend Algorithm~\ref{algo:planar} to genus $g$ graphs by dealing separately with local subgraphs of $G$ which are planar and which are not planar. For a graph $G$, the \emph{$T$-error set} of $G$ w.r.t. planarity is the vertex set defined by $X = \set{u\in V(G) \mid G[N^T[u]] \text{ is not planar}}$. 
\begin{theorem}\label{th:bounded_genus_simple}
    There exists a function $C$ such that for any $g$, there is a $\RATIO$-approximation \LOCAL algorithm for \minDS in Euler genus-$g$ graphs with round complexity $C(g)$. 
\end{theorem}
\begin{proof}
    Suppose $G$ has Euler genus $g$ and let $f$ be a $2$-dimensional control function of the class of Euler genus-$g$ graphs. A $2$-hop in $X\subset V(G)$ is a sequence of vertices $x_1,\dots,x_k$ in $X$ such that consecutive vertices are at distance at most two in $G$.
    
    \begin{algorithm}
    \caption{The algorithm for dominating set on bounded genus graphs.}\label{algo:MDS_bdd_genus}
    \begin{algorithmic}[1]
    
    \Require A graph $G$ of genus $\leq g$.
    
    \Ensure A dominating set $D$ of $G$ such that $\card{D}\le\RATIO \cdot \MDS(G)$.
    
    \State If $N^{f(10)+5}[u]$ is planar, $u$ runs \cref{algo:planar} on the corresponding locally planar subgraph.
    \Statex Let $X$ be the set of undominated vertices after step $1$.
    
    \State  If $u\in X$, $u$ computes lexicographically smallest minimum set $D_{X'}$ in $G$ which dominates the set $X'\subseteq X$ of vertices reachable from $u$ with a $2$-hop in $X$. It adds $N[u]\cap D_{X'}$ to $D$.
    
    \end{algorithmic}
    \end{algorithm}
    
    \begin{claim}\label{claim:bdd_genus_locally_nice}
        The round complexity is some constant $C(g)$ depending only on $g$. 
    \end{claim}
    
    \begin{proofofclaim}
        It suffices to show that there is no $ 2$-hop of weak diameter $\geq 2(f(10)+6)g$ of vertices in $X$. Suppose such a $2$-hop on vertices of $X$ exists. Take two vertices $x$ and $y$ on the $2$-hop which are at least $2(f(10)+6)g$ apart. Note that for any $i<2(f(10)+6)g$ there exists a vertex on the $2$-hop at distance $i$ or $i+1$ from $x$. Hence there exist $u_1,\dots u_{g+1}\in X$ where $u_i$ is at distance $2(f(10)+6)(i-1)$ or $2(f(10)+6)(i-1)+1$ from $x$. Therefore, the sets $N^{f(10)+5}[u_1],\dots,N^{f(10)+5}[u_{g+1}]$ are pairwise disjoint and they are non-planar, so the genus of the graph is at least $g+1$ by additivity (cf.~\cite[Theorem~4.4.3]{MT01}). As the genus is at most $g$, this is a contradiction.
    \end{proofofclaim}
    
    The coloring of the asymptotic dimension gives us color classes $C_1,C_2,C_3$, such that:
    \begin{itemize}[noitemsep]
    \item $C_i = B^i_{1} \cup \dots \cup B^i_{j_i}$ with\footnote{We denote by $\diam_G(X)$ the largest distance in $G$ between any two vertices of $X$.} $\diam_G(B^i_{j})\leq f(10)$ for all $1\leq j\leq j_i$, and
    \item for any $j\neq \ell$, we have $B^i_j\neq B^i_\ell$. More precisely, $B^i_j$ and $B^i_\ell$ are at distance at least $11$. 
    \end{itemize}
    
    \paragraph{The Cutting.}
    
    Let $Y_G$ be a minimum dominating set of $G$.
    Let $G'=G[N[Y']]$ where $Y'=Y_G\cap N^{5}[B^i_j]$. Note that $G'$ contains $N^4[B^i_j]$  and  $|\MDS(G')|\leq |Y_G\cap N^{5}[B^i_j]|$. 
    
    \paragraph{The Analysis Outside of the Error Set.}
    
    Note that $G'$ either contains only vertices from $X$ or contains at least one vertex which is not in $X$ and is therefore planar. Hence, vertices in $B^i_j\setminus X$ choose at most  $\RATIOP \cdot|\MDS(G')|\leq \RATIOP\cdot |Y_G\cap N^{5}[B^i_j]|$ vertices as their neighbourhood of distance $5$ in $G$ and $G'$ is the same. Vertices not in $X$ choose at most $\RATIOP |Y_G|$ vertices in each color class. Across all color classes, this adds up to at most $3 \cdot \RATIOP |Y_G|$.
    \paragraph{The Analysis With Non-Empty Error Set.} Let $X_1,\dots, X_k$ be the $2$-hop components.
    Note that for two distinct $2$-hop components $X_i$ and $X_j$, the neighbourhoods $N[X_i]$ and $N[X_j]$ do not intersect. As $|D_{X_i}| \leq |Y_G\cap N[X_i]|$, the vertices in $X$ choose at most $|Y_G|$ vertices.
    
    In total, at most $3\cdot 205 |Y_G|+|Y_G|=\RATIO\cdot |\MDS(G)|$ vertices are added to $D$: the algorithm is a $\RATIO$-approximation.
\end{proof}

In fact, if Algorithm~\ref{algo:planar} gives an approximation ratio $\alpha$, then the above algorithm for genus $g$ graphs is a $(3\alpha+1)$-approximation algorithm. One might now ask whether any algorithm for MDS on planar graphs can be turned into one for genus $g$ graphs. We show that this is indeed possible. In order to do so, we need to show that the cutting step and its analysis work for any MDS algorithm for planar graphs. We delve into this deeper in the next section.

\section{Meta-Theorems}\label{sec:metatheorems}

In this section, we present our meta-theorems, which take as input a \LOCAL approximation algorithm for a graph class $\sC$ and returns a \LOCAL approximation algorithm for a broader class $\sD$ that locally looks like graphs in $\sC$, possibly with some local exceptions.

Meta-theorems using black-box \LOCAL algorithms might be difficult to state as there are some subtle issues with vertex identifiers (see for instance~\cite{FKP13,FGKS13,GHS13}). The \LOCAL model assumes the input graph $G$ has unique $O(\log{n})$-bit vertex labels, i.e., with identifiers polynomial in $G$. In particular, running a \LOCAL algorithm $\sfA$ on only a small part $H$ of $G$ might fail because: (1) vertex labels in $H$ may not be polynomial in $H$ anymore, and (2) algorithm $\sfA$ may require polynomial-size identifiers to work properly. Thus, \MetaI assumes that the black-box algorithm $\sfA$ does not require polynomial identifiers. This captures all order-invariant algorithms\footnote{A model where vertices have unique labels, but the output of $\sfA$ is invariant under relabelings preserving the label order~\cite{GHS13}.}. In contrast, \MetaII and \MetaIII allow $\sfA$ to use polynomial identifiers, but impose restrictions on the optimization problem. Due to space constraints, Meta-Theorem~II {\&} III are presented in \cref{polynomial_ids_required}.

\paragraph{From Locally-$\sC$ to Locally Nice Graphs.}

\def\sT{\mathscr{T}}
\def\sP{\mathscr{P}}

Given a graph class $\sC$, we say that a graph class $\sD$ is \emph{$T$-locally-$\sC$} if for all $G \in\sD$ and $u\in V(G)$, $G[N^T[u]]\in\sC$.

This definition can already capture many graph classes, such as graphs of girth at least $2k+2$: 
those are just graphs that are $k$-locally-$\sT$, where $\sT$ is the class of trees. Similarly, graphs embedded on a surface with edge-width\footnote{The length of shortest noncontractible cycle in the embedded graph.} at least $2k+2$ are $k$-locally-$\sP$, where $\sP$ is the class of planar graphs.
However, our meta-theorem captures broader classes of graphs as it can support local exceptions.

Given a graph class $\sD$ and $G \in\sD$, the \emph{$T$-error set} of $G$ w.r.t. $\sC$ is the vertex set defined by $X = \set{u\in V(G) \mid G[N^T[u]] \notin \sC}$. In other words, this is the set of vertices that are witnesses for $G$ not being $T$-locally-$\sC$. Clearly, if $G$ has no $T$-errors, that is $X = \emptyset$, then $G$ is $T$-locally-$\sC$. Given a $\rho$-local\footnote{A problem whose solution can be verified by checking all neighborhoods of radius $\rho$. See next page for a formal definition.} problem $\Pi$, the class $\sD$ is $T$-locally $\delta$-nice w.r.t. $\sC$ and $\Pi$, if the maximum weak diameter in any $G\in\sD$ of any connected component of $G[N^{2\rho}[X]]$ is at most $\delta$, where $X$ is the $T$-errors of $G$ w.r.t. $\sC$. 
This means that the set of $T$-errors can be covered by far-away bounded-radius balls of $G$.
By \cref{prop:genus_locally_nice}, bounded Euler genus graphs form a locally nice class of graphs w.r.t. $\sP$.

\paragraph{Uniform Approximation.}

Our meta-theorems are inspired by the intriguing property~\cite[Proposition~3.1]{BGPW25} that any approximate algorithm $\sfA$ for MDS is going to perform well in a target graph class $\sD$ that is locally-$\sC$, as long as $\sD$ has small asymptotic dimension, and $\sfA$ makes "good local choices" in $\sC$. More precisely, a \LOCAL algorithm $\sfA$ is a \emph{$k$-uniform} $\alpha$-approximation in $\sC$, if for all $G\in\sC$ and $S\subseteq V(G)$, $\card{\sfA(G)\cap S} \le \alpha \cdot \MDS(G,N^k[S])$, where $\sfA(G)$ is the set returned by $\sfA$ when run on $G$, and $\MDS(G,X)$ is the minimum size of a subset of $V(G)$ dominating vertices of $X$ in $G$. 
Unless explicitly said, we will assume for convenience that $k\ge 1$, as otherwise only very restricted graph classes $\sC$ can support $0$-uniform constant-approximation algorithms for MDS (cf. \cref{prop:0uniform}).

We can now state properly the fact that any $k$-uniform approximate algorithm that performs well for $\sC$, performs well for a locally-$\sC$ class of bounded asymptotic dimension. It can be shown that the bounded asymptotic dimension assumption cannot be removed, even for MDS\footnote{This can be achieved by observing that a simple constant-round $3$-approximation \LOCAL algorithm exists for MDS in trees, but (high) girth-$g$ graphs do not admit a constant-round $o(g)$-factor approximation. The latter follows by combining Ramsey-theoretic arguments with the lower bound proof due to \cite{GHS13} based on $2k$-regular graphs of girth $g$, where vertices can be linearly ordered so that a positive fraction of them have pairwise isomorphic radius-$r$ neighbourhoods, for any $r,k,g$.}.

\begin{restatable}[{\cite[Proposition~3.1]{BGPW25}}]{proposition}{PropI}\label{prop:local_approx}
    Let $\sfA$ be a $k$-uniform $\alpha$-approximation \LOCAL algorithm for \minDS in a hereditary\footnote{A class of graphs closed under vertex deletion.} class of graphs $\sC$ with round complexity $r \ge 1$. Let $\sD$ be a graph class with $d$-dimensional control function $f$, and that is\footnote{The original statement claimed erroneously that $\sD$ being $t$-locally-$\sC$ for $t = f(2k+1)$ is enough, instead of the correct $t = f(2k+2) + r$. This has no impact, as in practice we use \cref{prop:local_approx} in \cref{th:locally_bounded_genus} only for $k \le r = O(1)$.} $(f(2k+2)+r)$-locally-$\sC$. Then $\sfA$ is also an $\alpha (d+1)$-approximation algorithm on $\sD$. 
\end{restatable}

We extend \cref{prop:local_approx} in multiple ways. First, we allow $\sD$ to be locally nice, and secondly, we prove that this also applies to problems other than MDS.
This means that, for example, any approximate algorithm $\sfA$ solving a well-behaved problem that performs well on planar graphs can be converted into an algorithm $\sfB$ that performs well on bounded genus graphs, as long as $\sfA$ is $k$-uniform on $\sC$, and $\sD$ has small asymptotic dimension. 
Before giving the full statement of \MetaI, we define the class of problems we are interested in.

\paragraph{\LOCAL Optimization Problems.}
\def\sfR{\mathsf{R}}

As defined in \cite{GHS13}, a \emph{simple graph problem} $\Pi$ is an optimization problem in which a feasible solution\footnote{Here, a feasible solution is any subset of vertices or edges that satisfies the constraints of the graph problem. Said differently, an admissible candidate solution.} is a subset of vertices or edges, and the goal is to optimize the size of a feasible solution.
A simple graph problem $\Pi$ is \textit{$\rho$-local} if there is a \LOCAL algorithm $\sfR$ with round complexity $\rho$ that can recognize a feasible solution\footnote{We do not require that the algorithm checks the optimality of the solution. If we require this, most interesting problems would require $\rho = \Omega(n)$ to be $\rho$-local.}. Algorithm $\sfR$ recognizes $S$ as a feasible solution if and only if, whenever applied on $G$ with $S$ given, $\sfR$ returns true for all vertices of $G$.
See \cref{sec:locally_verifiable_problems} for more formal definitions. 
We are mainly interested in small locality problems, that is when $\rho$ is constant, as clearly all simple graph problems are $\rho$-local for $\rho$ large enough.
A lot of well-known problems such as optimization variants of LCL problems~\cite{NS95a}.
This includes \textsc{Minimum Vertex Cover}, \minDS, \textsc{Maximum Independent Set}, \textsc{Maximal Matching}, and \minkTDS. The distance-$d$ versions of these problems\footnote{E.g. for \minDS, every vertex has a dominating vertex in its distance-$d$ neighborhood.} are also $d$-local problems. However, \minCDS is $\Omega(n)$-local\footnote{In a $n$-vertex cycle, no \LOCAL algorithm $\sfR$ of round complexity $< n/2$ can distinguish the situation where a feasible solution is $S = V(G)\setminus\set{u}$ from the situation where $S = V(G)\setminus\set{u,v}$ with two opposite vertices $u,v$.}.
Our framework includes problems such as \textsc{Red-Blue Dominating Set} that can take vertex or edge input labels on the graph, and problems defined on arbitrarily large degree graphs.

In this work, we focus on subsets of vertices, but the techniques presented here extend to subsets of edges.
In our setting, we restrict our attention to cases where the vertices do not have access to $n$, the size of the graph. For some problems, such as \minDS, this restriction can be relaxed\footnote{In full generality, this applies to all problems for which, for every $n$, there exists an $n$-vertex graph within the input graph class that has bounded diameter and admits a constant-size solution. For MDS on planar graphs, we can take a depth-1 tree.}. However, we omit this case for brevity, since we are not aware of any constant-round algorithm for a $\rho$-local problem that requires $n$. As, in the \LOCAL model with constant number of rounds, the literature mostly studies minimization problems, we also focus on those. Observe that many classically interesting problems like \textsc{Maximum Independent Set} and \textsc{Maximum Matching} do not admit constant-time constant-ratio approximations (see~\cite{BH25} for references), even on the restricted graph class of paths\footnote{The rough intuition for this is that on paths, any constant-time algorithm run on a vertex far away from the end of the path sees the same neighborhood, and therefore will output the same result on every vertex or edge. The only way to comply with the problem description, unlike minization problem like MDS, is to not take any vertex in the independent set (or any edge in the maximum matching), at least in the middle of the path. This leads to an arbitrarily bad approximation ratio.}.

Given a $\rho$-local problem $\Pi$, a graph $G$ and $X\subseteq V(G)$, 
we define $\OPT_\Pi(G,X)$ as the optimal size of a subset of vertices of $G$ that is feasible on $X$ (and not necessarily for vertices outside $X$).
We naturally extend the definition of $k$-uniform algorithm as follows: a \LOCAL algorithm $\sfA$ is a \emph{$k$-uniform} $\alpha$-approximation for $\Pi$ in a class of graphs $\sC$, if for all $G\in\sC$ and $S\subseteq V(G)$, $\card{\sfA(G)\cap S} \le \alpha \cdot \OPT_\Pi(G,N^k[S])$.

We can now state our strengthening of \cref{prop:local_approx}. Note that we further require (when the set of errors is not empty) that our problem is \emph{additive}, that is, any superset of a solution feasible on $X$ is still feasible on $X$.
For instance, MDS is additive because adding vertices to any dominating set maintains a dominating set.

\begin{restatable}[Meta-Theorem~I]{theorem}{MetaTheoremNonPoly}\label{th:metatheorem_nonpoly}
    Let $\sfA$ be a $k$-uniform $\alpha$-approximation \LOCAL algorithm that does not require polynomial identifiers, for an additive $\rho$-local problem $\Pi$ in a class of graphs $\sC$ with round complexity $r$, where $\sC$ is hereditary and stable by disjoint union.
    Let $\sD$ be a graph class of asymptotic dimension $d$ with $d$-dimensional control function $f$ that is $T$-locally $\delta$-nice w.r.t. $\sC$ and $\Pi$, where $T = f(2k+2\rho) + \max\set{k+\rho,r}$. Then, there exists a $\max\{k,\rho\}$-uniform $(\alpha (d+1)+1)$-approximation algorithm $\sfB$ on $\sD$ with round complexity $T+\delta+2$. Furthermore, if $\sD$ is $T$-locally $\sC$, the approximation ratio is $\alpha (d+1)$ and the additivity condition can be dropped.
\end{restatable}

\paragraph{Cutability.}

It may be difficult to prove that some algorithm is $k$-uniform.
Instead, we provide a simple sufficient criterion for deciding if \emph{any} constant-round constant-ratio approximation algorithm for $\Pi$ is $k$-uniform.

A $\rho$-local problem $\Pi$ is \emph{cuttable} for a hereditary graph class $\sC$ with parameter $\beta>0$ if for all $G\in \sC$ and $X\subseteq V(G)$, there exists some $Y \supseteq X$ such that $\OPT_\Pi(G[Y]) \le \beta \cdot \OPT_\Pi(G,X)$.
Informally, this asks that $X$ can be augmented into a set $Y$ (if not the whole vertex set) so that the graph induced by this bigger set $Y$ has an optimal solution somewhat smaller than the optimal feasible solution for $X$.  

For instance, MDS is cuttable with parameter $\beta = 1$ (see \cref{claim:cutting_trick_size}), as is the distance-$d$ version of MDS (see \cref{prop:MDS_distance_d_cuttable}).
We prove in \cref{appendix:cutting_implies_uniformizable_nonpoly} the fundamental property that every constant-round constant-ratio approximation for a cuttable problem is also $k$-uniform (as long as the algorithm does not require polynomial identifiers). Therefore, MDS admits uniform approximation algorithms. This means we can actually use \hyperref[th:metatheorem_nonpoly]{Meta-Theorem~I} for any cuttable problem without the $k$-uniform assumption.
Actually, for MDS, we can show that this property holds even if the algorithm requires polynomial identifiers (cf. \cref{prop:mds_uniformizable}).

Recall that Euler genus-$g$ graphs have asymptotic dimension at most $2$, and are locally nice (cf. \cref{prop:genus_locally_nice}). Thus, \MetaI implies that any constant-round $\alpha$-approximation for MDS gives a $C(g)$-round $(3\alpha+1)$-approximation for graphs of Euler genus $g$. 

The best known ratio for MDS in planar graphs is $\alpha = 11 + \eps$ due to Heydt et al.~\cite{HKOSV25}. Therefore,

\begin{restatable}{corollary}{BoundedGenus}\label{th:bounded_genus}
    For any $\eps>0$, for any $g$, there is a $k(\eps)$-uniform $(\RATIOG +\eps)$-approximation \LOCAL algorithm for \minDS in Euler genus-$g$ graphs with round complexity $C(\eps,g)$ for some functions $C$ and $k$.
\end{restatable}

We can even go one step further by applying \MetaI to the uniform approximation \LOCAL algorithm derived from~\cref{th:bounded_genus}, and the class $\sC$ of Euler genus-$g$ graphs. We immediately obtain the following, where $34.01$ can be thought of as $34+\varepsilon$.

\begin{restatable}{corollary}{LocallyBoundedGenus}\label{th:locally_bounded_genus}
    Let $\sC$ be the class of Euler genus-$g$ plus graphs, and $\sD$ be a graph class of asymptotic dimension $d$ that is $r$-locally-$\sC$ for some $r$ large enough. Then, there is a $34.01 (d+1)$-approximation \LOCAL algorithm for \minDS in $\sD$ with round complexity $C(d,g,r)$ for some function $C$.
\end{restatable}

Actually, \cref{th:bounded_genus,th:locally_bounded_genus} hold for the larger class of Euler genus-$g$ plus $a$-apex graphs, i.e. graphs obtained from a graph of Euler genus $g$ where at most $a$ vertices are allowed to have an arbitrary neighborhood. These graphs exclude $K_{3+a,2g+3+a}$ as minor (as, without apices, the graphs exclude $K_{3,2g+3}$), and thus have asymptotic dimension~2 as well. They are also locally nice, since $a$ apices cannot create short-cuts longer than $2a$. So, the approximation ratio is preserved, whereas the round complexity depends on $a$.

Moreover, our results can be applied to get approximation algorithms on graphs of bounded Euler genus for \emph{all} other minimization problems in the Distributed Computing literature that admit a constant-round, constant-ratio approximation on planar graphs (for a comprehensive list of such problems, see the survey of Feuilloley~\cite{Feuilloley23}). 

For instance, let us focus on the \minkTDS problem (studied extensively in~\cite{WAF02}), a generalization of \minDS. In this variant, the goal is to find the smallest subset of vertices such that every vertex outside the set is adjacent to at least $k$ vertices within it. The classical MDS problem corresponds to the special case $k = 1$.
For planar graphs, this problem admits a constant-round $k/(k-2)$-approximation algorithm~\cite{czygrinow2014distributed} for $k>2$, and a $6$-approximation~\cite{CZYGRINOW20171} for $k=2$. Both algorithms do not require polynomial identifiers.

Using the fact that \minkTDS is cuttable and additive (cf. \cref{prop:minkTDS_cuttable}), that Euler genus-$g$ graphs are locally nice (cf.  \cref{prop:genus_locally_nice}), and by applying \cref{prop:weak_meta_theorem} and \MetaI, we get the following:

\begin{restatable}{corollary}{OtherProblems}\label{th:otherproblems}    
    For every $k\ge 2$, there is a $C_1(k)$-uniform $\alpha$-approximation \LOCAL algorithm for \minkTDS in Euler genus-$g$ graphs with round complexity $C_2(g,k)$, for some functions $C_1,C_2$, where $\alpha = 19$ if $k=2$ and $\alpha \le 3k/(k-2) + 1$ otherwise.
\end{restatable}

For \minCDS~\cite{WAF02}, which is an example of non-local problem,
we can obtain a constant-round approximation with ratio $O(\sqrt{g})$ in Euler genus-$g$ graphs by applying the "connected" reduction from MDS in bounded expansion graphs due to~\cite[Theorem 5.8]{AOdMx18}. This gives an approximation ratio of $2\alpha \nabla_1 \le 68\sqrt{3g/2} + O(1)$, where $\alpha = 34+\eps$ is the best approximation ratio for MDS in Euler genus-$g$ graphs (\cref{th:bounded_genus}), and $\nabla_1 \le \sqrt{3g/2}+3$ (\cref{prop:nabla1}) is the expansion of Euler genus-$g$ graphs. We leave open the question of improving the ratio to a constant independent of the genus.

\newpage

\begin{acks}
    The first, third and fourth authors have been partially supported by the French ANR projects ENEDISC (ANR-24-CE48-7768) and TEMPOGRAL (ANR-22-CE48-0001). 
    The second author was supported in part by the Research Council of Finland, Grants 363558 and 359104.
    The sixth author was supported by the Deutsche Forschungsgemeinschaft (DFG, German Research Foundation) under Germany's Excellence Strategy - The Berlin Mathematics Research Center MATH+ (EXC-2046/1, project ID:390685689).
\end{acks}

\balance
\bibliographystyle{my_alpha_doi}
\newcommand{\etalchar}[1]{$^{#1}$}

\balance

\newpage
\appendix

\section{When the black-box algorithm requires polynomial identifiers}
\label{polynomial_ids_required}

As we have seen, \cref{prop:weak_meta_theorem} applies to lots of optimization problems. However, it does not fulfil the polynomial identifiers requirement of the \LOCAL model. We would like to have the same statement to apply in the \LOCAL model with polynomial identifiers with the same assumptions. However, assigning new identifiers to vertices while maintaining control of the output of the algorithm is very challenging. Ramsey-type techniques (for example, for finding lower bounds on the approximation ratio on particular graph classes~\cite{HLS14}) can be useful for this, but only on very particular graphs: ones where each neighborhood is similar. In general graphs, to the best of our knowledge, we do not know how to adapt this technique.
Thankfully, for locally-$\sC$ classes, where the set of errors is empty, we can weaken the statement of \cref{prop:weak_meta_theorem} in a way that we can still apply a variation of \MetaI afterwards.
The trick is to modify the coloring given by the asymptotic dimension so that all color classes are almost balanced in size (see \cref{lem:balance_size}).
With this technique, we just need to apply uniformity on sets that are large.
Formally, \LOCAL algorithm $\sfA$ is an $\eps$-weakly \emph{$k$-uniform} $\alpha$-approximation for $\Pi$ in a class of graphs $\sC$, if for all $G\in\sC$ and $S\subseteq V(G)$ such that $|S|\geq \big\lfloor\eps|V(G)|\big\rfloor$, we have $\card{\sfA(G)\cap S} \le \alpha \cdot \OPT_\Pi(G,N^k[S])$.
We say that the $d$-local problem $\Pi$ is \emph{$\eps$-weakly uniformizable} on $\sC$ if for every integer $d > 0$ and $r$ and any real $\alpha$, there exists some integers $k, r'$ and real $\beta$ such that for every $r$-round $\alpha$-approximation $\sfA$ of $\Pi$ on $\sC$ in constant time, there exists some $r'$-round \LOCAL algorithm $\sfB$ that is a $\eps$-weakly $k$-uniform $\beta$-approximation of $\Pi$ on $\sC$ in constant time.
The function $(r,\alpha)\mapsto (k,\beta)$ is called the (weak) \emph{binding} function.

It is a corollary of the proof of \cref{prop:weak_meta_theorem} that any cuttable problem is weakly uniformizable, in the \LOCAL model this time (see \cref{cor:cutting_trick_implies_weakly_uniformizable}).
Therefore, we can use \cref{lem:balance_size} to only work on color classes of linear size (which preserves the polynomial-ness of identifiers), and prove the following.

\begin{restatable}[Meta-Theorem~II]{theorem}{MetaTheoremNonPaddable}\label{th:metatheorem_nonpaddable}
    Let $d,k$ be integers $\alpha > 0$ a real, $\sfA$ be a $\frac{1}{d+1}$-weakly $k$-uniform $\alpha$-approximation \LOCAL algorithm for a $\rho$-local problem $\Pi$ in a class of graphs $\sC$ with round complexity $r$, where $\sC$ is hereditary and stable by disjoint union.
    Let $\sD$ be a graph class of asymptotic dimension $d$ with $d$-dimensional control function $f$, that is $T$-locally-$\sC$ for $T = f(6k+6\rho) + 4k+4\rho + \max\set{k+\rho,r}$. Then $\sfA$ is a $k$-uniform $\alpha (d+1)$-approximation algorithm on $\sD$.
\end{restatable}
The proof of this can be found in \cref{appendix:non-paddable_metatheorem}.

When the error set is not empty, we require further assumptions. For a discussion on why those assumptions are required, we refer to \cref{appendix:assumptions_meta_theorem_paddable}
A problem $\Pi$ is called \emph{$\omega$-paddable} if for any integer $N$, there exists some $G\in\sC$ on at least $N$ vertices such that $\OPT_\Pi(G)\leq \omega$.
We also add a small property that the problem needs to satisfy: a $d$-local problem $\Pi$ is \emph{dense} in $\sC$ if there exists some unbounded function $g$ that for any $G\in\sC$ of diameter $D$, we have $\OPT_\Pi(G)\geq \beta\cdot D$. We require that\footnote{\label{footnote:remove_dense}It is possible to get rid of this property if one adds that the graphs in the definition of $\omega$-paddable also have bounded diameter. Then, the rough idea is that, in the proof of \MetaIII, the algorithm will be able to compute the optimal solution on the graph $\widehat{G}$ of the proof. Then one gets \cref{eq:omega} without the $\omega$ term.} $\liminf f = +\infty$. 
This essentially mean that there are no large (in diameter) portion of the graph with no vertices of the solution; every vertex is close to a vertex of the solution.
Under the both additional assumptions of paddability and density, we can prove an analogue of \MetaII. See \MetaIII for the full statement.

\paragraph{Summary Table of the Results.}

Here is a simplified table of the assumptions necessary to apply our meta-theorem that transform a constant-round constant-factor approximation algorithm $\sfA$ on the class $\sC$ to another constant-round constant-factor approximation algorithm on the class $\sD$. We assume that $\sC$ is hereditary and stable by disjoint union, and that $\sD$ is of bounded asymptotic dimension.

\begin{table}[H]
    \begin{center}
    \begin{tabular}{|c|c|c|}
        \cline{2-3}
        \multicolumn{1}{c|}{}& $T$-error set is empty & $T$-error set is not empty \\
        \multicolumn{1}{c|}{}& ($\sD$ is locally-$\sC$) & ($\sD$ is locally nice) \\
        \hline
        & $\Pi$ is cuttable & $\Pi$ is cuttable and additive \\
        $\sfA$ \emph{does not assume}&&\\
        polynomial IDs & (\MetaI and & (\MetaI and \\
        & \cref{prop:weak_meta_theorem}) & \cref{prop:weak_meta_theorem}) \\
        \hline
        & $\Pi$ is cuttable & $\Pi$ is cuttable, additive,\\
        $\sfA$ \emph{assumes} & & paddable and dense\footnotemark{}\\
        polynomial IDs&&\\
        &(\MetaII and & (\MetaIII and \\
        & \cref{cor:cutting_trick_implies_weakly_uniformizable}) & \cref{cor:cutting_trick_implies_weakly_uniformizable}) \\
        \hline
    \end{tabular}
    \caption{Simplified table of the assumptions necessary to apply \MetaI, \MetaII and \MetaIII, along with their references.} 
    \end{center}
\end{table}
\footnotetext{It is possible to remove the dense condition, at the expense of making the paddable condition stricter. See~\cref{footnote:remove_dense}.}

\section{Discussions}

\subsection{Definitions for Locally Verifiable Problems}
\label{sec:locally_verifiable_problems}

Let us define $\rho$-local verifiable problems in more detail. Let $G$ be a graph with a set $\Sigma_{\text{in}}$ of input labels on vertices or edges, which includes a unique identifier for each vertex.
If $v\in V(G)$, $\sfA$ is a \LOCAL algorithm, and $X$ is a set of vertices of $G$, we use $\sfA(G, X, v)$ to represent the output of vertex $v$ when $\sfA$ is run on $G$ with input labels $\Sigma_{\text{in}}\sqcup X$ (the local inputs form an encoding of $X$, i.e. every vertex knows if it is contained or not in $X$).
A simple graph problem $\Pi$ is $\rho$-local if there exists some \LOCAL algorithm $\sfR$ which runs in $\rho$ rounds such that for all subsets of vertices $X$, the following holds:
\begin{itemize}
    \item if $X$ is a feasible solution of the problem $\Pi$ on $(G,\Sigma_{\text{in}})$, then for all vertices $v$ of $G$, we have $\sfA(G, X, v) = 1$, or
    \item if $X$ is not a feasible solution of the problem $\Pi$ on $(G,\Sigma_{\text{in}})$, then there exists $v\in V(G)$ such that $\sfA(G, X, v) = 0$.
\end{itemize}

\subsection{Why cutability seems unavoidable}

Cutability is not a trivial property. 
Consider the following very artificial problems: we want to output a set (as small as possible) which contains every vertex of degree $0$ or $1$ in the graph.
Of course, this problem admits a trivial $1$-round algorithm: take all vertices of degree $0$ or $1$.
However, our goal is to be able to use algorithms in a black-box manner.
Instead, we consider the $1$-round algorithm $\sfA$ that outputs every vertex of degree not equal to $2$.
Take $\sC$ the class of forests, on which $\sfA$ is a $2$-approximation.
This is because there are at least as many leaves as vertices of degree $3$ or more.

Now, let $\sD$ be the closure (under subgraphs) of $q$-subdivisions of ladders\footnote{The $q$-subdivision of a graph $G$ is a copy of $G$ where every edge is replaced by a path with $q$ inner vertices.}, where $q$ is as large as we like and a ladder is the Cartesian product of an edge and an arbitrarily long path. Clearly, every graph in $\sD$ is $q$-locally in $\sC$. Additionally $\sD$ has asymptotic dimension $1$, (one can check that they are $K_4$-minor-free, and so have bounded treewidth hence asymptotic dimension $1$~\cite{BBEGx24}).
However, $\sfA$ fails catastrophically on $\sD$.
Indeed, there are graphs in $\sD$ that do not have any leaves or isolated vertices (one can just take the $q$-subdivision of a ladder), so that the best solution on this graph is empty.
Of course, in this very specific scenario there does exist an alternative algorithm which gives a $1$-approx in forests and actually works in $\sD$. However, the point is to avoid modifying the black-box $\sfA$.
More generally, such degree-based problems seem to have uncontrollable behavior across local-global boundaries: in $\sC$ the algorithm can afford to take more vertices than necessary because of the existence of some other vertices (such as the large amount of leaves in the case of forests), whereas in $\sD$ these vertices do not exist.

\subsection{Cutability of \texorpdfstring{\textsc{Distance-Exactly-$d$ Dominating Set}}{Distance-Exactly-d Dominating Set}} 

In some graph classes such as planar graphs and minor-free classes, \cref{prop:MDS_distance_d_cuttable} can be extended to \textsc{Minimum Distance-Exactly-$d$ Dominating Set}.
Given a graph $G$ and $d$ a positive integer, a subset of vertices $D$ is a called a \emph{distance-exactly-$d$ dominating set} of $G$ if every vertex $v$ of $V(G)\setminus D$ is at distance exactly $d$ from some vertex in $D$. In the \textsc{Minimum Distance-Exactly-$d$ Dominating Set} problem, one asks for \emph{distance-exactly-$d$ dominating set} of minimum size.
In graph classes that are closed under the addition of leaves to the graphs, \textsc{Distance-$(\le\!d)$ Dominating Set} is cuttable with parameter $d$.
To give a rough intuition, in these classes, one can attach a path of length $d-1$ to a vertex $v$, put $v$ and all vertices in the path to the dominating set, and thus simulate a distance-$(\le\!d)$ dominating set that way.

\subsection{On the additional assumptions of the Meta-Theorem for paddable problems}\label{appendix:assumptions_meta_theorem_paddable}

In the statement of \MetaIII, we require problems that are paddable, i.e., we want to exclude graph classes that admit arbitrarily large graphs whose difficulty is concentrated in a bounded-diameter region with bounded optimum. Otherwise, a padding trick can inflate instance size (hence identifier ranges) without increasing hardness. 
One cannot hope for a version of \MetaIII with non-empty error sets and for problems without this property.
Here is the scenario we want to avoid. One could take any approximation algorithm on a class of graph $\sC$, and construct a class $\sD$ where graphs of size $n$ are formed by taking a graph of $\sC$ of size at most $\sqrt{n}$ (or $n^c$ for some $0<c<1$), and attaching somewhere an error set $X$ of size at least $n-\sqrt{n}$, which has small radius and small (bounded) optimum solution.
Then, this version of \MetaIII would give a good approximation on this class $\sD$. However, solving the problem $\Pi$ efficiently on $\sD$ would require solving it on $\sC$, as the error set is of small radius and has a small optimum solution.
Put in a nutshell, there is roughly a reduction from the class $\sD$ to the class $\sC$.
However, as the error set is large in size, one could attribute very large identifiers to the portion of the graph in $\sC$: setting $|X|\geq n-\sqrt{n}$ can artificially inflate the identifiers by a square (and one can get a polynomial with power $C$ by choosing $|X|\geq n^{1/C}$).
If the algorithm does not need the polynomial identifiers, everything still works. If it is not the case however, we have transformed an easy instance to another easy instance where identifiers are larger. This is not possible if the algorithm truly needs the polynomial identifiers, i.e., if the approximation guarantee depends on the vertex identifiers.\footnote{For these reasons, this might hint towards considering the size of the largest identifier as complexity parameter of its own, just like the size of the graph~\cite{balliu2026distributedcomplexitylandscapetrees}.} Thus, we require $\sC$ to be paddable or $\sD$ to be not paddable. We prove the result for the case where $\sC$ is paddable only, as the proof is similar for the other case.
\footnote{If we require the algorithm to have knowledge the size of the graph, we can have \MetaIII with this additional property. However, to achieve this, we need to change the definition of paddable: we now request, for any $n$, a graph of size $n$ (instead of size $\geq n$) with the desired properties.}

\section{Proofs}

\subsection{Approximation for planar graphs}
\label{appendix:planar}
\ApproxPlanar*

\begin{proof}
\begin{figure}[h!]
\centering
\begin{tikzpicture}[scale=1]

\node[vertex,draw=mygreen] (v1)  at (4.5,3.5)   {1};
\node[vertex,draw=myred,fill=myred!30!white, ultra thick] (v2)  at (4.5,2.5) {2};
\node[vertex,line width=1pt,draw=primary] (v3)  at (3,3)   {3};
\node[vertex,line width=1pt,draw=primary] (v4)  at (3,2)   {4};
\node[vertex,draw=mygreen] (v5)  at (5.5,2) {5};
\node[vertex,line width=1pt,draw=primary] (v6)  at (5,1)   {6};
\node[vertex,line width=1pt,draw=primary] (v7)  at (4,1.8) {7};
\node[vertex,line width=1pt,draw=primary] (v8)  at (4,0.8) {8};

\node[vertex, draw=mygreen] (v10) at (1.5,2) {10};
\node[vertex, draw=mygreen] (v9)  at (2.5,1.2) {9};
\node[vertex,line width=1pt, draw=primary] (v11) at (1,1) {11};
\node[vertex,line width=1pt, draw=primary] (v12) at (1.8,0.3) {12};
\node[vertex,line width=1pt, draw=primary] (v13) at (3,0.3) {13};
\node[vertex,line width=1pt, draw=primary] (v14) at (0.5,0.2) {14};
\node[vertex,line width=1pt, draw=primary] (v15) at (0,1) {15};
\node[vertex,line width=1pt, draw=primary] (v16) at (-0.5,0.2) {16};
\node[vertex,line width=1pt, draw=primary] (v17) at (-1,1) {17};
\node[vertex,line width=1pt, draw=primary] (v18) at (-1.5,0.2) {18};
\node[vertex,line width=1pt, draw=primary] (v19) at (-2,1) {19};
\node[vertex,line width=1pt, draw=primary, draw=none] (v20) at (-2.5,0.2) {$\dots$};

\draw[edge] (v1)--(v2)--(v5)--(v6)--(v8)--(v7)--(v2);
\draw[edge] (v3)--(v1);
\draw[edge] (v3)--(v4)--(v7);
\draw[edge] (v4)--(v2);
\draw[edge] (v7)--(v5);
\draw[edge] (v8)--(v4);

\draw[edge] (v10)--(v4);
\draw[edge] (v10)--(v9);
\draw[edge] (v9)--(v4);
\draw[edge] (v9)--(v8);
\draw[edge] (v9)--(v12);
\draw[edge] (v10)--(v11);
\draw[edge] (v11)--(v12);
\draw[edge] (v11)--(v14);
\draw[edge] (v12)--(v14);
\draw[edge] (v11)--(v15);
\draw[edge] (v9)--(v13);
\draw[edge] (v12)--(v13);
\draw[edge] (v15)--(v16);
\draw[edge] (v15)--(v17);
\draw[edge] (v14)--(v16);
\draw[edge] (v14)--(v15);
\draw[edge] (v17)--(v18);
\draw[edge] (v17)--(v19);
\draw[edge] (v16)--(v18);
\draw[edge] (v16)--(v17);
\draw[edge] (v19)--(v20);
\draw[edge] (v18)--(v20);
\draw[edge] (v18)--(v19);
\end{tikzpicture}
\caption{Depicted in green is the dominating set $D_u$ for $u=2$. Note that $d(u)$ is $1$ or $5$. }\Description{~}
\end{figure}
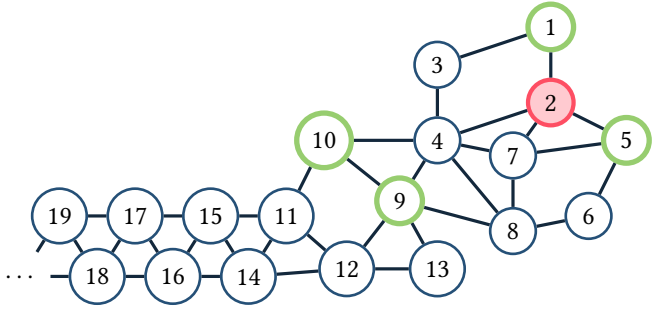

Let $G$ be a planar graph, let $Y$ be a nice dominating set of $G$ and let $D$ be the dominating set of $G$ returned by the algorithm. Let $u\in V(G)$ and let $D_u$ be the dominating set computed by $u$. That means
\begin{itemize}
    \item $D_u$ is nice
    \item $D_u$ is lexicographically smallest nice set dominating $N^4[u]$
\end{itemize}

The first key observation is that given disjoint subsets  $A,B,C$ of $V(G)$ such that  $G = G[A \cup C]\cup G[B\cup C]$, that is, $C$ separates $A$ from $B$ (see Figure~\ref{fig:ABC}). If $G[B\cup C]$ is in the fourth neighbourhood of $u \in B\cup C$, the value of $D_u \cap B$ is completely determined by $D_u \cap C$.
In other words, there is a unique best way to extend $D_u \cap C$ to $B$. Therefore, for such vertices $u$, there are at most $2^{|C|}\cdot MDS(G,B\cup C)$ different candidates for $d(u)$.

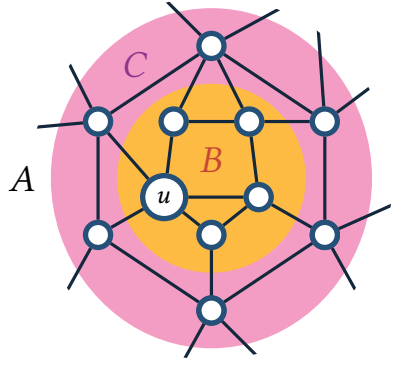
\begin{figure}[htbp!]
    \centering
\begin{tikzpicture}[scale=0.5,use Hobby shortcut,every node/.style={vertex,draw=primary}]
        \draw [fill=mypink,draw=none] ([closed]3, 24) .. (7, 21) .. (7, 18) .. (3, 15) .. (-1, 18) .. (-1, 21);

        \draw [fill=myorange,draw=none] ([closed]1, 21) .. (3, 22) .. (5, 21) .. (5, 18) .. (3, 17) .. (1, 18);

        \node [rectangle,draw=none,fill=none,text=Red!80!black] (42) at (3, 20) {\LARGE\bm $B$};
		\node [rectangle,draw=none,fill=none,text=Fuchsia] (43) at (1, 22.5) {\LARGE\bm $C$};
        \node [rectangle,draw=none,fill=none] (44) at (-2, 19.5) {\LARGE\bm $A$};

        \clip (3, 19.5) circle (4.8cm);
        
        \node (1) at (3, 16) {};
		\node (2) at (0, 18) {};
		\node (3) at (6, 18) {};
		\node (4) at (6, 21) {};
		\node (5) at (3, 23) {};
		\node (6) at (0, 21) {};
		\node (7) at (-1, 23) {};
		\node (8) at (-0.25, 25.5) {};
		\node (9) at (8.25, 19) {};
		\node (10) at (5, 14) {};
		\node (11) at (-2.5, 15.75) {};
		\node (12) at (8, 22.5) {};
		\node (13) at (5.75, 24.5) {};
		\node (14) at (1, 25) {};
		\node (15) at (-2.75, 20.75) {};
		\node (16) at (-3.25, 23.5) {};
		\node (17) at (-2.25, 24.75) {};
		\node (18) at (7.25, 15.75) {};
		\node (19) at (6, 18) {};
		\node (20) at (3, 16) {};
		\node (21) at (0, 14.25) {};
		\node (22) at (-1.25, 15.75) {};
		\node (23) at (1.75, 13.75) {};
		\node (24) at (2, 21) {};
		\node (25) at (4, 21) {};
		\node (26) at (3, 18) {};
		\node (28) at (4.25, 19) {};
		\node (29) at (1.75, 19) {$u$};

        \draw[edge] (2) to (1);
		\draw[edge] (4) to (5);
		\draw[edge] (5) to (6);
		\draw[edge] (6) to (2);
		\draw[edge] (1) to (3);
		\draw[edge] (3) to (4);
		\draw[edge] (14) to (5);
		\draw[edge] (6) to (7);
		\draw[edge] (7) to (14);
		\draw[edge] (14) to (8);
		\draw[edge] (8) to (7);
		\draw[edge] (15) to (6);
		\draw[edge] (16) to (7);
		\draw[edge] (17) to (7);
		\draw[edge] (5) to (13);
		\draw[edge] (13) to (4);
		\draw[edge] (4) to (12);
		\draw[edge] (3) to (9);
		\draw[edge] (19) to (18);
		\draw[edge] (20) to (10);
		\draw[edge] (21) to (23);
		\draw[edge] (23) to (22);
		\draw[edge] (22) to (2);
		\draw[edge] (20) to (23);
		\draw[edge] (21) to (22);
		\draw[edge] (22) to (11);
		\draw[edge] (5) to (24);
		\draw[edge] (5) to (25);
		\draw[edge] (25) to (24);
		\draw[edge] (25) to (4);
		\draw[edge] (24) to (29);
		\draw[edge] (29) to (6);
		\draw[edge] (29) to (2);
		\draw[edge] (20) to (26);
		\draw[edge] (26) to (29);
		\draw[edge] (29) to (28);
		\draw[edge] (28) to (25);
		\draw[edge] (26) to (28);
		\draw[edge] (28) to (19);
    \end{tikzpicture}
    \caption{The depiction of $A,B,C$ in the proof of the key observation.}\Description{~}
    \label{fig:ABC}
\end{figure}

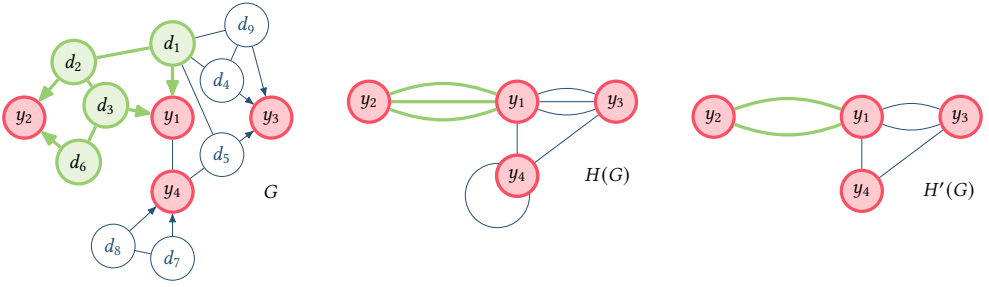
\begin{figure}
\centering
\resizebox{.95\textwidth}{!}{\begin{tikzpicture}[scale=0.5,every path/.style={draw=primary}] 
    \node (before) at (0,0){
        \begin{tikzpicture}[scale=0.43]  
        \def\height{1.5};
        \def\width{2};
        \node[snode] (y_2) at (0, 0){$y_2$};
        \node[snode] (y_1) at (3*\width,0){$y_1$};
        \node[snode] (y_3) at (5*\width,0){$y_3$};
        \node[snode] (y_4) at (3*\width,-3){$y_4$};
        \node[green node] (d_1) at (3*\width, 2*\height) {$d_1$};
        \node[green node] (d_3) at (2*\width-0.7,0.5){$d_3$};
        \node[green node] (d_2) at (1*\width,1.5*\height){$d_2$};
        \node[green node] (d_6) at (1*\width+0.2,-\height-0.3){$d_6$};
        \node[normal node] (d_4) at (4*\width,\height){$d_4$};
        \node[normal node] (d_9) at (4.5*\width,2.5*\height){$d_9$};
        \node[normal node] (d_5) at (4*\width, -1*\height){$d_5$};
        \node[normal node] (d_7) at (3*\width,-5.7){$d_7$};
        \node[normal node] (d_8) at (1.8*\width,-5.2){$d_8$};
        \draw[ultra thick,grun,{Latex[round]}-] (y_1) to (d_1);
        \draw[ultra thick,grun,{Latex[round]}-] (y_1) to (d_3);
        \draw[ultra thick,grun] (d_1) to (d_2);
        \draw[ultra thick,grun] (d_3) to (d_2);
        \draw[ultra thick,grun] (d_3) to (d_6);
        \draw[ultra thick,grun,{Latex[round]}-] (y_2) to (d_2);
        \draw[ultra thick,grun,{Latex[round]}-] (y_2) to (d_6);    
        \draw[{Latex[round]}-] (y_3) to (d_4);
        \draw[{Latex[round]}-] (y_3) to (d_5);
        \draw (d_1) to (d_4);
        \draw (d_1) to (d_5);
        \draw (d_1) to (d_9);
        \draw (d_4) to (d_9);
        \draw[{Latex[round]}-] (y_3) to (d_9);
        \draw(y_4) to (y_1);
        \draw (y_4) to (d_5);
        \draw (d_7) to (d_8);
        \draw[-{Latex[round]}] (d_7) to (y_4);
        \draw[-{Latex[round]}] (d_8) to (y_4);
        \node at (5*\width,-3){$G$};
        \end{tikzpicture}
    };
    \node(after) at (before.east)[anchor=east,xshift=6cm,yshift=-0.2cm]{
        \begin{tikzpicture}[scale=.43]
        \def\height{1.5};
        \def\width{2};
        \node[snode] (y_2) at (0, 0){$y_2$};
        \node[snode] (y_1) at (3*\width ,0){$y_1$};
        \node[snode] (y_3) at (5* \width, 0){$y_3$};
            \node[snode] (y_4) at (3*\width ,-3){$y_4$};
        \draw[ultra thick, grun, bend left=20] (y_2) to (y_1);
        \draw[ultra thick, grun,bend left=20] (y_1) to (y_2);
        \draw[bend left=20] (y_1) to (y_3);
        \draw[bend left=20] (y_3) to (y_1);
        \draw ($(y_4)+(-0.8,-0.8)$) circle (1.3cm);
        \draw (y_1) to (y_4);
        \draw[ultra thick,grun] (y_1) to (y_2);
        \draw (y_3) to (y_4);
        \draw (y_1) to (y_3);
        \node at (5*\width, -3){$H(G)$};
        \node[snode] (y_4) at (3*\width ,-3){$y_4$};
        \end{tikzpicture}
    };
    \node(after) at (before.east)[anchor=east,xshift=12cm,yshift=-0.2cm]{
        \begin{tikzpicture}[scale=.43]
        \def\height{1.5};
        \def\width{2};
        \node[snode] (y_2) at (0, 0){$y_2$};
        \node[snode] (y_1) at (3*\width ,0){$y_1$};
        \node[snode] (y_3) at (5* \width,0){$y_3$};
        \node[snode] (y_4) at (3*\width ,-3){$y_4$};
        \draw[ultra thick, grun, bend left=20] (y_2) to (y_1);
        \draw[ultra thick, grun,bend left=20] (y_1) to (y_2);
        \draw[bend left=20] (y_1) to (y_3);
        \draw[bend left=20] (y_3) to (y_1);
        \draw (y_1) to (y_4);
        \draw (y_4) to (y_3);
        \node at (5*\width,-3){$H'(G)$};
        \end{tikzpicture}
    };
\end{tikzpicture}
}
\caption{The plane graphs $H$ and $H'$ obtained from $G$.}\Description{~}
\label{fig:HH'G}
\end{figure}
We now discuss how to use that observation to reach the desired conclusion. We associate to each vertex $u \in V\setminus Y$ an arbitrary vertex $y_u \in Y \cap N(u)$. We consider the plane multigraph $H$ obtained from an arbitrary planar embedding of $G$ by contracting every edge $uy_u$ for $u \in V\setminus Y$, but keeping any resulting loop or multiple edge. Note that $V(H) = Y$ by construction. 
Let $H'$ be the plane multigraph obtained from $H$ by iteratively deleting any edge incident to a face bounded by one edge (i.e. a loop) or to two faces bounded by two edges (i.e. there is a multiedge of multiplicity $3$), see Figure~\ref{fig:HH'G}. In particular, $H'$ may contain loops and multiple edges, but in the embedding they separate different parts of $V(H')$ or, in the case of multiple edges, are incident on at least one side to a face of size at least $3$. We observe that by Euler's formula, there are fewer than $2\cdot 3|V(H')| = 6|Y|$ edges in $H'$. To every edge $e$ of $H$ or $H'$ we associate an edge $e_G$ of $G$ such that $e_G$ turns into $e$ in the contraction process. We let $Z$ be the subset of vertices of $V(G)\setminus Y$ which are an endpoint of $e_G$ for some edge $e\in E(H')$. Since every edge $e_G$ contains at most $2$ vertices of $Z$, we have $|Z|\le 12 |Y|$. Therefore, there are at most $13|Y|$ vertices of $Y\cup Z$ in $D$.

It remains to bound $|D\setminus (Y \cup Z)|$. Note that there are two types of vertices $v\in V(G) \setminus(Y\cup Z)$. 
\begin{itemize}
    \item $N[v]\subseteq N[y_v]$ 
    \item $N[v]\nsubseteq N[y_v]$ and $v$ is between two edges $e_G,e'_G$ in $G$ such that $e$ and $e'$ bound a face of size $2$ in $H'$
\end{itemize}
There cannot be any other case, as if the first case does not apply then $v$ has degree at least $2$, so is adjacent to some vertex $w$ such that $w$ is not a neighbour of $y_u$. But then $w$ gets contracted to a vertex $s\neq y_u$ when creating $H$. Note that $vw$ does not get contracted and let $e$ be the edge in $H$ such that $e_G=vw$. But since $v$ is not in $Z$, the edge $e$ gets deleted when creating $H'$, i.e. $e$ is in between two parallel edges to $e$ and we are in the second case.

We now assume $w$ is a vertex such that $d(w)=v$. If the first case applies then any neighbour of $v$ is also a neighbour of $y_v$, in particular, the distance of a vertex $w\neq v$ to $v$ is the same as to $y_v$. However, either $y_v$ has strictly more neighbours than $v$, or $y_v$ has a smaller label than $v$, as $y_v\in Y$, hence $w$ would prefer $y_v$ over $v$.
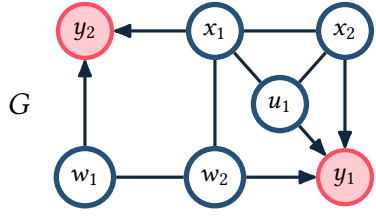
\begin{figure}
\centering
    \begin{tikzpicture}[scale=0.43,every path/.style={draw=primary}]  
    \def\height{1.5};
    \def\width{2};
    
    \node[snode] (y_3) at (5* \width, -1*\height){$y_1$};

    \node[vertex,draw=primary] (y_4) at (1*\width ,-1*\height){$w_1$};
    \node[vertex,draw=primary] (d_1) at (3*\width, 2*\height) {$x_1$};

    \node[snode] (y_1) at (1*\width, 2*\height){$y_2$};

    \node[vertex,draw=primary] (d_4) at (4*\width, 0.5*\height){$u_1$};
    \node[vertex,draw=primary] (d_9) at (5*\width, 2*\height){$x_2$};
    \node[vertex,draw=primary] (d_5) at (3*\width, -1*\height){$w_2$};
    
    \draw[edge,-{Latex[round]}](d_1) to (y_1);
    
    \draw[edge,{Latex[round]}-] (y_3) to (d_4);
    \draw[edge,{Latex[round]}-] (y_3) to (d_5);
    \draw[edge] (d_1) to (d_4);
    \draw[edge] (d_1) to (d_5);
    \draw[edge] (d_1) to (d_9);
    \draw[edge] (d_4) to (d_9);
    \draw[edge,{Latex[round]}-] (y_3) to (d_9);
    \draw[edge,-{Latex[round]}](y_4) to (y_1);
    \draw[edge] (y_4) to (d_5);
    \node at (0*\width, 0.5*\height){\Large$G$};
    \end{tikzpicture}
    \caption{A depiction of $x_1,x_2,w_1,w_2$ for the vertex $u_1$.}\Description{~}
    \label{fig:region}
\end{figure}

We analyse now the second case, which is slightly more interesting. Let $e_G$ be incident to $x_1,x_2$ and $e'_G$ to $w_1,w_2$, where for each $i$, either $d(x_i)=y_i$ or $x_i=y_i$ and similarly for $w_i$ and without loss of generality we assume $y_{u}=y_1$. For a depiction, see Figure~\ref{fig:region}. Note that $C=\{x_1,x_2,y_1,y_2,w_1,w_2\}$ separates $u$ from the rest of $V(H')$, and that vertices involved in some edge inside $f$ in $H$ are all adjacent to $y_1$ or to $y_2$ hence are at distance at most $3$ to all vertices in $f$ (with respect to $u$, since $u$ is a neighbour of $y_2$ or has a neighbour which is a neighbour of $y_2$ as we are in the second case). But this means that $w$ is at distance at most $4$ to all vertices in $f$. Therefore, there are at most $2^{6} $ choices for $D_u$ and  $2^{6}\cdot 2$ choices for $d(u)$ if $d(u)\notin Y\cup Z$. This is because not more than $2$ vertices from the set contained in the cycle $y_2x_1w_1y_1w_2x_2y_2$ can be chosen as otherwise $y_1,y_2$ would be a better choice. We can note that in the case where both $y_1$ and $y_2$ get selected in $D_u$ there is no extra vertex to add as all of the inside is dominated and when one of $y_1$ or $y_2$ is chosen, then there is at most one vertex picked on the inside, otherwise $y_1,y_2$ would be a better choice. Therefore, we reduce the number of choices to $1\cdot 2^4\cdot 2+ 2\cdot 2^4\cdot 1 = 64$.

Since there are at most $3|Y|$ options for $f$, this adds up to $192|Y|$.

All together, we get $13|Y|$ for the set $Y$ or the set $Z$ being chosen, and $192|Y|$ for the choices of the second case of vertices not in $Z$ or $Y$ being chosen. This sums up to $\RATIOP |Y|$, as claimed.
\end{proof}

\subsection{Meta-theorem for the \texorpdfstring{\LOCAL}{LOCAL} model with no polynomial identifiers requirement}

\MetaTheoremNonPoly*

This proof is essentially a simplified proof of \MetaIII. However, we include it for completeness.

\begin{proof}
    Let $G\in \sD$.
    As $\sD$ has $d$-dimensional control function $f$, then $G^{2k+2\rho}$ admits a $(d+1)$-coloring $f(2k+2\rho)$-bounded in $G$.
    Fix such a coloring, and, for each $i\in\set{0,1,\dots,d}$, let $C_i$ be the set of color-$i$ vertices in $G^{2k+2\rho}$ (and also in $G$). 
    We denote by $\CC_i$ the set of $(2k+2\rho)$-components of $C_i$ (i.e. connected components of $G^{2k+2\rho}[C_i]$). 
    So, by definition, all components of $\CC_i$ have weak diameter in $G$ at most $f(2k+2\rho)$, and are pairwise at distance at least $2k+2\rho+1$ in $G$ (because distinct components of $\CC_i$ cannot be adjacent in $G^{2k+2\rho}$).

    We describe the Algorithm~$\sfB$ we want to use in \cref{algo:PiI}, which is very similar to \cref{algo:MDS_bdd_genus}.
    \begin{algorithm}
    \caption{Generic algorithm $\sfB$ from \MetaI.}\label{algo:PiI}
    \begin{algorithmic}[1]

    \Require An $\eps$-weakly $k$-uniform $\alpha$-approximate \LOCAL algorithm $\sfA$ for $\Pi$ in $\sC$ with round complexity $r$, a graph classes $\sC$ and $\sD$ with properties of the statement of \MetaIII, the $d$-dimensional control function $f$ of $\sD$ and $G\in\sD$ with $T$-errors $X$ where $T = f(2k+2\rho) + \max\set{k+\rho,r}$.

    \Ensure A solution $S$ for the graph $G$ to the problem $\Pi$ with the guarantee that 
    $\card{S}\leq(\alpha (d+1)+1) \cdot \OPT_\Pi(G)$.

    \State $S \gets \emptyset$

    \State Each vertex $u$ computes $G[N^T[u]]$ and checks whether it belongs to $\sC$; as a consequence, it decides whether it belongs to $X$. 

    \State Each vertex $u$ runs $\sfA$ for $r$ rounds, and gets added to $S$ only if $u\in \sfA(G)\setminus X$.

    \State $S \gets S \cup S'$, where $S'$ is a brute-forced minimum set of $G$ that satisfies the constraints of vertices that are non-satisfied.

    \end{algorithmic}
    \end{algorithm}

    We run Algorithm~$\sfB$ on $G$. It is easy to check that $\sfB(G)$ is a valid solution for $G$, because $\Pi$ is additive. By definition, $\card{\sfB(G)} \le |S| + |S'|$, where $S = \sfA(G)\setminus X$ (Step~3) and $S'$ is a brute-forced set of minimum size satisfying the constraints of vertices of $G$ not satisfied (Step~4). Note that every vertex not in $N^{\rho}[X]$ is satisfied by $S$, so $S'\subseteq N^{2\rho}[X]$. Clearly, $\card{S'} \le \OPT_\Pi(G)$, and $\card{S'} = 0$ if there is no $T$-errors ($X=\emptyset$). Moreover, by definition of the coloring, $\card{\sfA(G)\setminus X} = \sum_{i=0}^d \card{(\sfA(G)\cap C_i)\setminus X}$. In other words,
    \begin{equation}\label{eq:B(G)nonpoly}
    \card{\sfB(G)} ~\le~ \sum_{i=0}^d \card{(\sfA(G)\cap C_i)\setminus X} + \left\{\begin{array}{ll}
        0 &\mbox{if $X=\emptyset$}\\
        \OPT_\Pi(G) &\mbox{otherwise}
        \end{array}
        \right.
    \end{equation}
    
    To upper bound the term $\card{(\sfA(G)\cap C_i)\setminus X}$, consider some color $i\in\set{0,1,\dots,d}$ and some component $C\in \CC_i$ such that $C\not\subseteq X$. Consider some $v_C\in C\setminus X$, and denote $G_C = G[N^T[v_C]]$. Note that, by definition of $X$ and as $v\notin X$, we have $G_C \in\sC$. Moreover, let $\CC'_i = \{C\in\CC_i \mid C\not\subseteq X\}$ and let $G' = G[\bigcup_{C\in\CC'_i} N^T[v_C]]$ be the union of the $G_C$'s for $C\in\CC_i$ such that $C\not\subseteq X$.

    Along the proof, we will use several times the following fact:

    \begin{restatable}{fact}{DLocalProblem}\label{fact:d_local_problem}
        For any graph $H$ and $W \subseteq V(H)$, to satisfy the constraints of vertices of $W$ in $H$, it is enough to select vertices taken from $N^\rho[W]$.
    \end{restatable}
    
    In the following, to differentiate the neighborhoods in $G$ and in $G_C$ or $G'$, we write them using either $N_G$, $N_{G_C}$, or $N_{G'}$ notations.
    Notice that, for each non-negative integer $t\le \max\set{k+\rho,r}$, we have $G[N_{G}^t[C]]\subseteq G_C$. This is because $C$ has weak diameter in $G$ at most $f(2k+2\rho)$ and thus $G[N_{G}^t[C]]$ has weak diameter in $G$ at most $f(2k+2\rho) + t \le T$, by the choice of $T$. Moreover, as $\sC$ is hereditary, we have that $G[N_{G}^t[C]]\in \sC$.

    As $G[N^t_G[C]]\subseteq G_C$, vertices in $C$ have the same distance $t$-neighborhood in $G$ and $G_C$.
    In particular, $G[N^t_G[C]] = G_C[N^t_{G_C}[C]]$ and $G[N^t_G[\bigcup\CC'_i]] = G'[N^t_{G'}[\bigcup\CC'_i]]$.
    From~\cref{fact:d_local_problem}, any minimum set of vertices of $G$ that satisfies the constraints of $N^k_G[C]$ is contained in $G[N^{k+\rho}_G[C]]$, and similarly if we replace $G$ by $G_C$.
    It follows that $\OPT_\Pi(G,N^k_{G}[C]) = \OPT_\Pi(G_C,N^k_{G_C}[C])$. 
    Recall that any two connected components $C, C' \in \CC_i$ are at distance in $G$ at least $2k+2\rho+1$. So $G[N^{k+\rho}_G[C]]$ and $G[N^{k+\rho}_G[C']]$ are disjoint subgraphs, and, as to satisfy all constraints of vertices in $V(G)$, one needs to satisfy constraints of vertices in $N^{k}_G[C]$ and $N^{k}_G[C']$, by disjunction of these sets and by~\cref{fact:d_local_problem}, we have by summing over all $C\in\CC_i$ that
    \[
        \OPT_\Pi\left(G,N^k_{G}\left[\bigcup \CC'_i\right]\right) = \OPT_\Pi\left(G',N^k_{G'}\left[\bigcup \CC'_i\right]\right).
    \]
    Moreover, because $\sC$ is stable by disjoint union and that $G[N_{G}^t[C]]\in\sC$ for every $C\in \CC'_i$, we get $G[N_{G}^t[\bigcup\CC'_i]]\in\sC$. Therefore, as $\sfA$ is a $k$-uniform $\alpha$-approximation on $\sC$, we get
    \begin{align*}
            \card{\sfA(G')\cap \bigcup\CC'_i}\leq \alpha\cdot\OPT_\Pi\left(G',N_{G'}^k\left[\bigcup\CC'_i\right]\right) = \OPT_\Pi\left(G,N^k_{G}\left[\bigcup \CC'_i\right]\right).
    \end{align*}

    Since $G[N^r_G[\bigcup\CC'_i]] = G'[N^r_{G'}[\bigcup\CC'_i]]$, Algorithm~$\sfA$, which runs in $r$ rounds, returns the same vertex set in $G$ and $G'$ for vertices in $\bigcup\CC'_i$. Therefore, $\card{\sfA(G)\cap \bigcup\CC'_i} = \card{\sfA(G')\cap \bigcup\CC'_i}$. Combining with the previous inequality, we get

    \begin{align*}
        \card{\sfA(G)\cap \bigcup\CC'_i} ~\le~ \alpha \cdot \OPT_\Pi\left(G, N_G^k\left[\bigcup\CC'_i\right]\right)~\le~ &\alpha \cdot \OPT_\Pi(G,V(G))~=~ \alpha \cdot \OPT_\Pi(G)~.
    \end{align*}
    We remark that $\card{(\sfA(G)\cap C)\setminus X} = 0$ if $C \subseteq X$. Because all vertices in $C_i$ but not in $X$ are in some $C\in\CC'_i$, we have
    
    \begin{align}\label{eq:nonpoly_Ci_bound}
        \card{(\sfA(G)\cap C_i)\setminus X} \leq \card{\sfA(G)\cap \bigcup\CC'_i} \leq \alpha \cdot \OPT_\Pi(G)~.
    \end{align}

    Putting \cref{eq:nonpoly_Ci_bound} in \cref{eq:B(G)nonpoly}, we get that Algorithm~$\sfB$ is an $(\alpha(d+1)+1)$-approximation, and even an $\alpha(d+1)$-approximation if $X=\emptyset$.
    
    To see that $\sfB$ is also a $k$-uniform approximation, it is sufficient to see that, for any $W \subseteq V(G)$, the bound obtained is $\card{\sfB(G)\cap W} \le \alpha (d+1)\cdot \OPT_\Pi(G, N_G^k[W]) + \OPT_\Pi(G, N_G[W])$ where the first two terms are due to running $\sfA$ and the last term is due to running the brute-force.
    Indeed, if the brute-force computed a set whose intersection with $W$ was smaller than $\OPT_\Pi(G, N^\rho_G[W])$, we could replace it by the minimum size set satisfying the constraints of the vertices in $N^\rho_G[W]$ and obtain a smaller set, a contradiction. This proves that $\sfB$ is a $\max\{k,\rho\}$-uniform with an approximate ratio $\alpha(d+1)+1$ (or $\alpha(d+1)$ if $X=\emptyset$, as the brute-force is not required in that case).
    Now, let us prove that Algorithm~$\sfB$ has the desired round complexity.
    \begin{itemize}
        \item Computing $X$ in Step~2 takes $T+1$ rounds.
    
        \item Running $\sfA$ in Step~3 takes~$r$ rounds.
    
        \item For Step~4, consider the set $S$ computed by $\sfB$ at Step~3, before the brute-force.
        Observe that the vertices in the set $W = V(G)\setminus N^\rho_G[X]$ are satisfied by $S$. This is because vertices of $W$ are at distance at least $\rho+1$ from $X$ and thus vertices of $N^\rho_G[W]$ cannot be in $X$. Thus, $\sfA$ applies to all vertices of $N^\rho_G[W]$, and so $W$ is indeed satisfied by $S$ (\cref{fact:d_local_problem}).
        Therefore, to satisfy $N^\rho_G[X]$ it is enough to select a set $S'$ from $N^{2\rho}_G[X]$ (\cref{fact:d_local_problem}) as done in Step~4.
        By assumption, the connected components of $N_G^{2\rho}[X]$ have weak diameter in $G$ at most $\delta$ (obviously, if $X = \emptyset$ then $\delta = 0$). Therefore, the brute-force (Step~5) will take at most $\delta+1$ rounds. 
    \end{itemize}
    Naively, Steps~2 and~3 together take $(T+1) + r$ rounds. As Algorithm~$\sfA$ is not guaranteed to work when executed on vertices of $X$, Step~3 must be run only after Step~2. However, both steps can be run in parallel as follows. We run $\sfA$ for $r$ rounds exactly and stop it just after (since its running time on a vertex of $X$ could result into more than $r$ rounds). Then, the decision to add $u$ in $S$ (if selected by $\sfA$) is delayed for the next $T+1 - r$ rounds (this is non-negative from the choice of $T$). In parallel, $G[N^T[u]]$ is computed and is checked to be in $\sC$ or not after $T+1$ rounds. So, after $\max\set{r,T+1} = T+1$ rounds, set $S$ in Step~3 has been completed.
    Step~4 takes $\delta+1$ steps, so that the total round complexity of $\sfB$ is $T+\delta+2$, completing the proof. 
\end{proof}

\subsection{\texorpdfstring{\minDS}{Minimum Dominating Set} is uniformizable}\label{appendix:MDS_uniformizable}

\begin{restatable}{proposition}{MDSUniformizable}\label{prop:mds_uniformizable}
    Let $\sC$ be a hereditary class of unbounded degree\footnote{For a graph class in which every graph has maximum degree $\Delta$, the $\LOCAL$ algorithm that includes all nodes is a $\Delta+1$-approximation. This implies that the bounded-degree case is less interesting than the unbounded-degree one, which is therefore the focus of our study.} and that is closed by disjoint union.
    Then, the \minDS problem is uniformizable in the \LOCAL model with binding function $(r,\alpha)\mapsto(r+1,\alpha+\eps)$ for any real $\eps>0$.
\end{restatable}

\begin{proof}
    We first prove this statement in the \LOCAL model with $\eps = 0$.
    Assume for sake of contradiction that there exists integers $d,r$ and a real $\alpha$ such that for all integers $k$ and real $\beta$, there exists an $r$-round $\alpha$-approximation $\sfA$ of $\Pi$ in constant time on the hereditary graph class $\sC$ such that $\sfA$ is not $k$-uniform with $\beta$-approximation of $\Pi$ in constant time.
    Here, we will take $k=r+1$ and $\beta=\alpha$.
    By the definition of uniformity, this means that there exists some subset $S$ of vertices for which:
    \[
        |\sfA(G)\cap S| > \alpha \cdot \MDS(G,N^k[S]).
    \]
    Let us cut $G$ around $S$: set $G' := G[N^k[S]\cup X]$ where $X\subseteq N^{k+1}[S]$ is a minimum dominating set of $N^k[S]$ in $G$.
    We keep the same identifiers in $G'$ as in $G$.
    We first bound the size of the minimum dominating set of $G'$.
    \begin{claim}\label{claim:cutting_trick_size}
        $\MDS(G')\leq \MDS(G,N^k[S])$.
    \end{claim}
    \begin{proofofclaim}
        By definition of $X$, it follows that $|X|=\MDS(G,N^k[S])$.
        $X$ is also a dominating set of $G'$, because each of its vertices is either in $X$ or in $N^k[S]$, and that in the second case, it is dominated by some vertex of $X$ by definition of $X$.
        Therefore, $\MDS(G')\leq |X| = \MDS(G,N^k[S])$.
    \end{proofofclaim}
    As our goal is to get a contradiction on $G'\in\sC$~\footnote{because $\sC$ is hereditary} instead of $G$, we prove the following.
    \begin{claim}\label{claim:cutting_trick_view}
        On $S$, algorithm $\sfA$ returns the same set in $G$ and in $G'$.
    \end{claim}
    This will prove that $|\sfA(G)\cap S| = |\sfA(G')\cap S|$.
    \begin{proofofclaim}
        $\sfA$ is an $r$-round algorithm. Therefore, vertices in $S$ only see vertices and edges in $G[N^r[S]]$, and edges between vertices of $N^r(S)$ and vertices of $N^{r+1}(S)$~\footnote{The definition of an $r$-round \LOCAL algorithm running on vertex $v$ includes edges at distance $r+1$ of $v$ but not vertices at distance $r+1$ of $v$ in $G$.}. All of those are the same in $G$ and $G'$, therefore $\sfA(G)\cap S = \sfA(G')\cap S$, as wanted.
    \end{proofofclaim}
    By \cref{claim:cutting_trick_view},
    \[
        |\sfA(G')|\geq|\sfA(G')\cap S| = |A(G)\cap S| > \alpha\MDS(G,N^k[S]).
    \]
    By \cref{claim:cutting_trick_size}, we have $\alpha\MDS(G,N^k[S])\geq \alpha\cdot\MDS(G')$.
    It follows that $|\sfA(G')|\geq \alpha\cdot\MDS(G')$, a contradiction as $G'\in \sC$ and $\sfA$ is an $\alpha$-approximation on $\sC$.
    Let us now adapt this proof for the \LOCAL model, for $\eps > 0$. For now, by taking $\gamma=\alpha+\eps$, we have proved $|\sfA(G')|\geq (\alpha+\eps)\cdot\MDS(G')$. The only thing we need to take care of is that nodes in $G$ have identifiers that are bounded by a polynomial of $|V(G')|$.
    Instead of changing the identifiers in $G'$ without impacting the output of $\sfA$, which is tricky, we consider a bigger graph.
    Let $N$ be the biggest identifier in $G'$. $\sC$ is a hereditary graph class of unbounded degree, so there exists a graph $H\in\sC$ and a vertex $v$ of $H$ with degree at least $N$.\footnote{Here, if the algorithm requires access to $n$ (the size of the graph), one needs to choose $N\geq n$ and delete neighbors of $v$ so that, in the following, we get $|V(G'')|=n$.}
    Take $H'=H[N_H[v]]$ and give the vertices identifiers in $[N+1,2N]$. Let $G'' = G'\sqcup H'$.
    Clearly, $H'$ can be dominated by choosing the single vertex $v$, i.e. $\MDS(G'') = \MDS(G') + 1$.
    Furthermore, $\sfA(G'') = \sfA(G') \sqcup \sfA(H')$ and it follows
    \[
        |\sfA(G'')|\geq (\alpha+\eps)\cdot\MDS(G')+1 = (\alpha+\eps)\cdot\MDS(G'')+1-(\alpha+\eps).
    \]
    Now, we can apply the same trick as in \cref{prop:aMDS+b}. 
    Observe that the diameter of each connected component of the graph $G'$ is at most $3\MDS(G')$.
    We construct an algorithm $\sfB$ where each vertex checks whether the diameter $D$ of its component is less than $3B$, for $B =(\alpha-1)/\eps$. This can be done by collecting its radius-$3B$ neighborhood in $3B$ rounds. If true, the vertex locally brute forces an optimal solution for its component, which it already knows about. If false, the vertex applies the approximate algorithm $\sfA$.
    If $D\leq B$, then the algorithm is obviously a good approximation (with approximation ratio $1$).
    Otherwise, $D>3B$ and it follows that $\MDS(G')>B$. One can verify that $|\sfA(G'')| > \alpha\cdot \MDS(G'')$, a contradiction, as $G''\in \sC$.
\end{proof}

\subsection{\texorpdfstring{\minkTDS}{Minimum k-Tuple Dominating Set} is cuttable and additive}

\begin{proposition}\label{prop:minkTDS_cuttable}
    \minkTDS is additive and cuttable with parameter $1$.
\end{proposition}

\begin{proof}
    By definition, the problem is additive. Let us now prove cutability.
    Let $G$ be a graph and $X\subseteq V(G)$, and let $S$ be a $k$-Tuple Dominating Set of the set $X$.
    Set $Y=X\cup S$. $S$ is a $k$-Tuple Dominating Set of $G[Y]$. 
    Indeed, every vertex of $G[Y]$ that is in $Y\setminus S$ is in $X\setminus S$, and therefore is $k$-Tuple Dominated by $S$.
    Thus for $\Pi= \minkTDS$, we have $\OPT_\Pi(G,X)\geq \OPT_\Pi(G[Y])$.
\end{proof}

\subsection{\texorpdfstring{\textsc{Distance-$\leq\! d$ Dominating Set}}{Distance-≤d Dominating Set} is cuttable}\label{appendix:distance-d_DS}

\begin{restatable}{proposition}{DistanceDCuttable}\label{prop:MDS_distance_d_cuttable}
    \textsc{Distance-$(\le\!d)$ Dominating Set} is cuttable with parameter $1$.
\end{restatable}

\begin{proof}
    Let $\Pi$ be the \textsc{Distance-$(\leq\! d)$ Dominating Set} problem, let $G$ be a graph and $X$ a subset of vertices. 
    Let $D$ be a minimum-size set that is distance-$d$ dominating for the vertices in $X$, so that $|D|=\OPT_\Pi(G,X)$.
    Note that $D\subseteq N^d[X]$.
    We are going to construct a set $X\subseteq Y\subseteq N^d[X]$ such $\OPT_\Pi(G[Y]) \leq \OPT_\Pi(G,X)$.
    We do this recursively. Start with $Y := X$.
    We maintain the invariant that every vertex of $Y$ that is not dominated by $D\cap Y$ in $G[Y]$ is in $X$.
    This invariant holds trivially for $Y=X$.
    While some vertex $v$ of $X$ is not dominated in $G[Y]$ by $D\cap Y$, do the following.
    Let $y\in D$ be a vertex that dominates $v$ in $G$, i.e., with the distance to $y$ at most $d$.Note that here, we do not have necessarily $y\notin Y$: it could be the case that $v$ is dominated by some $y\in X$ in $G$ because of a path with a portion outside $X$.
    As $y\in D\subseteq N^d[X]$, there is a path $P$ in $G[N^d[X]]$ from $v$ to $y$ with length at most $d$. Actually, one can just take the shortest $vy$-path. It has length at most $d$ and is in $N^d[v]\subseteq N^d[X]$.
    Set $Y := Y\cup V(P)$. All vertices that are in $P$ (including $v$) are dominated by $y$, and therefore are now dominated in $G[Y]$.
    The invariant is thus satisfied.
    Each time we do this operation, at least one more vertex of $X$ is dominated in $G[Y]$. As $X$ is finite, this process terminates and we have constructed a set $Y$ that is dominated by $D\cap Y$ in $G[Y]$.
    It follows that $\OPT_\Pi(G[Y]) \leq \OPT_\Pi(G,X)$, and this completes the proof.
\end{proof}

\subsection{Cutting implies uniform algorithms in the \LOCAL model with no polynomial identifiers requirement}\label{appendix:cutting_implies_uniformizable_nonpoly}

A problem $\Pi$ is \emph{uniformizable} on $\sC$ if for all integers $r$ and all reals $\alpha$, there exists some integers $k$ and real $\beta$ such that every $r$-round $\alpha$-approximation $\sfA$ of $\Pi$ on $\sC$ is a $k$-uniform $\beta$-approximation of $\Pi$ on $\sC$.
The function $(r,\alpha)\mapsto (k,\beta)$ is called the \emph{binding} function.

\begin{restatable}{proposition}{WeakMetaTheorem}\label{prop:weak_meta_theorem}
    Let $\sC$ be a hereditary class that is closed by disjoint union, and let $\Pi$ be a $\rho$-local problem which is cuttable with parameter $\beta$.
    Then, $\Pi$ is uniformizable \emph{in the \LOCAL model with no polynomial identifiers requirement} with binding function $(r,\alpha) \mapsto(r+1,\alpha\cdot\beta)$.
\end{restatable}

\begin{proof}
    Assume for sake of contradiction that there exist an integer $r$ and a real $\alpha$ such that for all integers $k,r'$ and real $\gamma$, there exists an $r$-round $\alpha$-approximation $\sfA$ of $\Pi$ on the hereditary graph class $\sC$ but no $r'$-round $k$-uniform $\gamma$-approximation algorithm $\sfB$ for $\Pi$ on $\sC$.
    Here, we take $k=r+1$, $\gamma=\alpha\cdot\beta$, $r'=r$ and $\sfB=\sfA$, i.e. suppose $\sfA$ is not a $k$-uniform $\gamma$-approximation of $\Pi$.
    By the definition of uniformity, this means that there exists some subset $S$ of vertices for which:
    \[
        |\sfA(G)\cap S| > \alpha\beta \cdot \OPT_\Pi(G,N^k[S]).
    \] 
    As $\Pi$ is cuttable with the particular case $X=N^k[S]$, there exists a set $Y$ such that $\OPT_\Pi(G[Y])\leq \beta \cdot \OPT_\Pi(G,N^k[S])$.
    We now study the graph $G' := G[Y]$, with the same identifiers as in $G$. We prove the following.
    \begin{claim}\label{claim:cutting_trick_view_2}
        We have $\sfA(G)\cap S = \sfA(G')\cap S$.
    \end{claim}
    \begin{proofofclaim}
        This is the same proof as in \cref{claim:cutting_trick_view}, except with a set $Y$ instead.
        As $\sfA$ is an $r$-round algorithm, vertices in $S$ only see vertices and edges in $G[N^r[S]]$, and edges between vertices of $N^r(S)$ and vertices of $N^{r+1}(S)$. Those are the same in $G$ and $G'$ as $N^{r+1}[S] = N^k[S] \subseteq Y$, and so we have $\sfA(G)\cap S = \sfA(G')\cap S$, as wanted.
    \end{proofofclaim}
    By \cref{claim:cutting_trick_view_2},
    \[
        |\sfA(G')|\geq|\sfA(G')\cap S| = |A(G)\cap S| > \alpha\beta\cdot\OPT_\Pi(G,N^k[S]).
    \]
    By the cutability, we have $\OPT_\Pi(G')\leq \beta \cdot \OPT_\Pi(G,N^k[S])$.
    It follows that $|\sfA(G')|> \alpha\cdot\OPT_\Pi(G')$, a contradiction as $G'\in \sC$ and $\sfA$ is an $\alpha$-approximation on $\sC$.
\end{proof}

\subsection{The algorithm of Heydt et al. is uniform}

\begin{restatable}{observation}{AlgoIsUniform}\label{obs:localplanar}
    The algorithm described in~\cite[Theorem~2.3]{HKOSV25} is a $k(\eps)$-uniform $(11+\eps)$-approximation \LOCAL algorithm for \minDS in planar graphs with round complexity $C(\eps)$, for every $\eps>0$, and for some functions $C$ and $k$. 
\end{restatable}
\begin{proof}
    Let us consider the algorithm described in \cite[Theorem~2.3]{HKOSV25}.
    It has a number of rounds bounded by a function of $\eps$, and does not need polynomial identifiers so \cref{obs:localplanar} follows by \cref{prop:mds_uniformizable}.
\end{proof}

For the interested reader, here is a more detailed analysis of the running time of the algorithm.
It returns a solution $D_1\cup D_2\cup D_3^1 \cup D_3^2$.
$D_1$ can be computed in two rounds, $D_2$ in two additional rounds, and $D_3^1$ in one more additional round.
The set $D_3^2$ can be computed in time $O(\log \Delta/\log(1+\eps))$ where $\Delta = 4\nabla_1\cdot(4^{\nabla_1}+2\nabla_1)(\Delta_R+1)/\eps$ with (with the particular case of planar graphs) $\nabla_1 < 3$ and $\Delta_R=\kappa^{s-1}(t+s-1+(s-1)\kappa^3)$ with $\kappa = \max\{2\nabla_0,2\nabla\}$.
Again, for planar graphs, $\nabla_0<3$ and $\nabla \leq 2$ so $\kappa \leq 6$.
It follows that $\Delta_R\leq 15732$ because $s=t=3$ on planar graphs.
So $\Delta \leq 13593312/\eps = O(1/\eps)$ and the running time is $O(\log(1/\eps)/\log(1+\eps/2)) = O(\log(1/\eps)/\eps)$.

\subsection{Cutability implies weakly uniform algorithms}\label{appendix:cutability_implies_weakly_unformizable}

\begin{restatable}{corollary}{CutabilityImpliesWeaklyUnif}\label{cor:cutting_trick_implies_weakly_uniformizable}
    Let $\sC$ be a hereditary class closed by disjoint union, and let $\Pi$ be a $\rho$-local problem which is cuttable with parameter $\beta$.
    Then, for any real $\eps > 0$, the problem $\Pi$ is $\eps$-weakly uniformizable \emph{in the \LOCAL model} with (weak) binding function $(r,\alpha)\mapsto(r+1,\alpha\cdot\beta)$.
\end{restatable}

\begin{proof}
    Let us go through the proof of \cref{prop:weak_meta_theorem} one more time. For sake of brevity, we do not copy the whole identical proof, but instead give a summary of the proof.
    Assume for sake of contradiction that there exists integers $d,r$ and a real $\alpha$ such that for all integers $k$ and real $\gamma$, there exists an $r$-round $\alpha$-approximation $\sfA$ of $\Pi$ in constant time on the hereditary graph class $\sC$ such that $\sfA$ is not $\eps$-weakly $k$-uniform with $\gamma$-approximation of $\Pi$ in constant time.
    Here, we will take $k=r+1$ and $\gamma=\alpha\cdot\beta$.
    By the definition of weak uniformity, this means that there exists some subset $S\subseteq V(G)$ of size at least $\eps |V(G)|$ for which:
    \[
        |\sfA(G)\cap S| > \alpha\beta \cdot \OPT_\Pi(G,N^k[S])~.
    \]
    We have showed there exists a set $Y$ such that $\beta \cdot \OPT_\Pi(G,N^k[S])\geq \OPT_\Pi(G[Y])$.
    Let $G' := G[Y]$. We keep the same identifiers in $G'$ as in $G$.
    We have showed that $|\sfA(G')|\geq \alpha\cdot\OPT_\Pi(G')$, which would be a contradiction with the fact that $\alpha$-approximation on $\sC$, if $\sfA$ was an algorithm in the \LOCAL model with no polynomial identifiers requirement.
    For this to be a contradiction in the \LOCAL model, we have to show that the IDs in $G'$ are also polynomial in $|V(G')|$.
    Now, we just use the additional property that $S$ has large size.
    All the identifiers in $G'$ are polynomial in $|V(G)|$, say smaller than $D\cdot |V(G)|^c$ for some real $c>0$.
    As $|V(G')|\geq|S|\geq \eps\cdot|V(G)|$, the identifier in $G'$ are smaller than $D\cdot |V(G')|^c/\eps^c$, which is still polynomial. This finishes the proof.
\end{proof}

\subsection{Bounded Euler genus graphs are locally nice}

\begin{proposition}\label{prop:genus_locally_nice}
    Let $\Pi$ be a $\rho$-local problem. Then, for any $g\geq 1$, the class of Euler genus-$g$ graphs is $T$-locally $\delta$-nice w.r.t. planar graphs and $\Pi$, where $\delta < g\cdot(2T+4\rho+1)$.
\end{proposition}

\begin{proof}
To derive a contradiction, consider a connected component $H$ of $G[N^{2\rho}[X]]$, and assume that $H$ has weak diameter $\delta \ge g \cdot (2T+4\rho+1)$ in $G$. As $g\ge 1$, $X$ and $H$ are not empty, and $H$ is not planar. So, $H$ must contain some path $P$ such that the distance in $G$ between its endpoints, say $s$ to $t$, is $\delta$. In particular, for each $d \in\set{0,1,\dots,\delta}$, there must exist a vertex of $P$ that is at distance exactly $d$ in $G$ from~$s$. This is because the distance in $G$ from $s$, when moving along an edge of~$P$, is a function that can vary by at most one unit, and this distance goes from $0$ (at $s$) to $\delta$ (at $t$). For every $i \in\set{0,1,\dots,g}$, let $u_i$ be any vertex of $P$ at distance exactly $d_i =  i\cdot (2T+4\rho+1)$ in $G$ from $s$. So $u_0 = s$, $u_g = t$, and all the $u_i$'s exist in $P$ (so in $H$) since $d_i \in\set{0,1,\dots,\delta}$. By definition of $H$, each $u_i$ intersects $N^{2\rho}[X]$. So, for each $u_i$ one can select a vertex $x_i\in X$ at distance in $G$ at most~$2\rho$ from $u_i$. By the triangle inequality, for all $0\le i<j\le g$, $x_i$ and $x_j$ are at distance in $G$ at least $(d_j-2\rho) - (d_i+2\rho) = (j-i) \cdot(2T+4\rho+1) - 4\rho \ge 2T+1$. It follows that $N^T[x_i]$ and $N^T[x_j]$ are disjoint. By definition of $X$ and $H$, the subgraphs $G[N^T[x_i]]$'s are not planar, thus of Euler genus $\ge 1$. By additivity of the Euler genera\footnote{The plural of genus.}
    of these $g+1$ pairwise disjoint subgraphs of $G$ (cf.~\cite[Theorem~4.4.3]{MT01}), we get a contradiction that $G$ has Euler genus $g$. Thus $\delta < g\cdot (2T+4\rho+1)$ as claimed.
\end{proof}

\subsection{Balanced asymptotic dimension}\label{appendix:balanced_asdim}

\begin{restatable}{lemma}{BalanceSize}\label{lem:balance_size}
    Let $\sG$ be a graph class with asymptotic dimension at most $d$ and control function $f$.
    For every graph $G\in\sG$ with $n$ vertices and every integer $r\geq 1$, there exists
    a cover $C_0,C_1,\dots,C_d$ of $V(G)$ such that
    \begin{itemize}
        \item every $r$-component of each $C_i$ has weak diameter in $G$ at most $f(3r)+2r$, and
        \item for every $i\in\range{0}{d}$ we have $|C_i|\ge \floor{\frac{n}{d+1}}$.
    \end{itemize}
\end{restatable}

\begin{proof}
    Let $q:=\big\lfloor\frac{n}{d+1}\big\rfloor$ and $r$ be any integer.
    By \cref{prop:asdim_def} applied to scale $3r$, there exists a $(d+1)$-coloring of $G^{3r}$ whose monochromatic components have weak diameter at most $f(3r)$ (as they have weak diameter $f(3r)$ in $G^{3r}$) in $G$.
    Let $C_0,C_1,\dots,C_d$ be the corresponding color classes, with $|C_0| \leq |C_1|\leq \cdots\leq |C_d|$ without loss of generality. They form a cover of $V(G)$, and every $3r$-component of each $C_i$ has weak diameter at most $f(3r)$ in $G$.

    We construct sets $C'_0,\dots,C'_d$ by proving the following by induction of $i$.
    For every $i$, there exists sets $C'_0,\dots,C'_{i-1}$ such that
    \begin{enumerate}[label=(I\arabic*),ref=(I\arabic*)]
        \item\label{item:I1} for every $j\leq i$, we have $|C'_j|\geq q$, and if $|C'_j|>q$ then $C'_j = C_j$,
        \item\label{item:I2} every $r$-component of each $C'_j$ (for $j\leq i$) has weak diameter at most $f(3r)+2r$ in $G$, and
        \item\label{item:I3} $\displaystyle\bigcup_{j\leq i} C'_j \supseteq \bigcup_{j\leq i} C_j$.
    \end{enumerate}
    Note that \ref{item:I3} says that the sets $C'_0,\dots,C'_i$ cover at least every vertex of $C_0,\dots,C_i$.
    For $i=0$ these conditions are trivially satisfied.
    Assume the conditions hold for some $i\in\{0,\dots,d\}$.
    Let $A:=N^r[C_i]\setminus C_i$ and $B:=(V(G)\setminus N^r[C_i])\setminus C_i$, such that $A\sqcup B=V(G)\setminus C_i$.
    Intuitively, we will try to add vertices to $U$ from $A$ whenever possible (these are within distance $r$ of the original $C_i$), and otherwise take vertices from $B$ (which are farther from $C_i$).
    
    If $|C_i|\geq q$, then take $C'_j = C_j$ for all $j\geq i$. \ref{item:I1} is satisfied because $q \leq |C_i|\leq \cdots\leq |C_d|$.
    \ref{item:I2} is satisfied because the $C_i$'s have all weak diameter at most $f(3r)$, and \ref{item:I3} is trivially satisfied.
    So assume $|C_i|<q$ and let $s:=q-|C_i|>0$ be the number of new vertices we must add to $C_i$ to reach size $q$.
    There are two possibilities.

    \begin{itemize}
        \item \textbf{Case 1: $|A|\geq s$.}
        Choose an arbitrary set $X\subseteq A$ of size $s$ and let $C'_i := C_i\cup X$; this satisfies \ref{item:I1} and \ref{item:I3}.
        Let us prove \ref{item:I2}.
        Let $x,x'\in A$ be two vertices in the same $r$-component of $C'_i$.
        As $C'_i\subseteq N^r[C_i]$, there exist $y,y'\in C_i$ such that $d(x,y),d(x',y')\leq r$.
        We first prove that $y$ and $y'$ are in the same $3r$-component of $C_i$.
        Because $x$ and $x'$ are in the same $r$-component of $C'_i$, exists a sequence $x=x_1,x_2,\dots,x_t=x'$ in $C'_i$ with $d(x_j,x_{j+1})\leq r$ for all $j$.
        As $C'_i\subseteq N^r[C_i]$, we can choose, for each $j$, some $y_j\in C_i$ such that $d(x_j,y_j)\leq r$. Without loss of generality, choose $y_1:=y$ and $y_t:=y'$.
        Then for any $j$, we have that
        \[
        d(y_j,y_{j+1})
        \leq d(y_j,x_j)+d(x_j,x_{j+1})+d(x_{j+1},y_{j+1})
        \leq 3r,
        \]
        and so $y$ and $y'$ are in the same $3r$-component of $C_i$, which means that $d(y,y')\leq f(3r)$.
        Hence
        \begin{align*}        
            d(x,x')&\leq d(x,y)+d(y,y')+d(y',x')\\ &\leq r+f(3r)+r=f(3r)+2r,
        \end{align*}
        as required. This proves \ref{item:I2}.

        \item \textbf{Case 2: $|A|< s$.}
        Then by definition of $B$, 
        \[
            |B| \geq |V(G)|-|A|-|C_i| \geq n-s+1-|C_i| = n-q+1.
        \]
        As the sets $C_j$ for $j\neq i$ cover $B$, there exists some $j\neq i$ such that $|C_j\cap B|\geq(n-q+1)/d \geq q$.  
        Choose an arbitrary set $X\subseteq C_j\cap B$ of size $s$ and let $C'_i := C_i\cup X$; this satisfies \ref{item:I1} and \ref{item:I3}.
        \ref{item:I2} is true all $r$-components of $C'_i$ are either $r$-components of $C_i$ (so they have diameter at most $f(3r)$) or are subsets of $r$-components of $C_j$ (and the same diameter bound applies).
    \end{itemize}

    In both cases we have constructed sets $C'_i$ that satisfies \ref{item:I1}, \ref{item:I2} and \ref{item:I3}.
    By \ref{item:I3}, $\bigcup_{i=0}^d C'_i= \bigcup_{i=0}^d C_i = V(G)$, so the $C'_i$ form a cover.
    \ref{item:I1} and \ref{item:I2} prove the two wanted conditions on the cover. This finishes the proof. 
\end{proof}

\subsection{Meta-theorem for non-paddable problems with empty error set}\label{appendix:non-paddable_metatheorem}

\MetaTheoremNonPaddable*

\begin{proof}
    Let $G\in \sD$ and set $\ell := f(6k+6\rho) + 4k+4\rho$.
    As $\sD$ has $d$-dimensional control function $f$, then by \cref{lem:balance_size}, one can assume $G^{2k+2\rho}$ admits a $(d+1)$-coloring $\ell$-bounded in $G$ where every color class has size at least $\big\lfloor\frac{n}{d+1}\big\rfloor$. Fix such a coloring, and, for each $i\in\set{0,1,\dots,d}$, let $C_i$ be the set of color-$i$ vertices in $G^{2k+2\rho}$ (and also in $G$), so that $|C_i|\geq\big\lfloor\frac{n}{d+1}\big\rfloor$. We denote by $\CC_i$ the set of $(2k+2\rho)$-components of $C_i$ (i.e. connected components of $G^{2k+2\rho}[C_i]$). 
    So, by definition, all components of $\CC_i$ have weak diameter in $G$ at most $\ell$, and are pairwise at distance at least $2k+2\rho+1$ in $G$ (because distinct components of $\CC_i$ cannot be adjacent in $G^{2k+2\rho}$).
    
    We run Algorithm~$\sfA$ on $G$. It is easy to check that $\sfA(G)$ is a valid solution for $G$, because $\sD$ is $T$-locally-$\sC$. By definition of the coloring, 
    \begin{align}\label{eq:nonpaddable}
        \card{\sfA(G)} = \sum_{i=0}^d \card{\sfA(G)\cap C_i}~.
    \end{align}
    
    To upper bound the term $\card{\sfA(G)\cap C_i}$, consider some color $i\in\set{0,1,\dots,d}$ and some component $C\in \CC_i$. Consider some $v_C\in C$, and denote $G_C = G[N^T[v_C]]$. Note that as $\sD$ is $T$-locally-$\sC$, we have $G_C \in\sC$. Let $G' = G[\bigcup_{C\in\CC_i} N^T[v_C]]$ be the union of the $G_C$'s for $C\in\CC_i$.

    Along the proof, we will use several times the following fact:

    \DLocalProblem*
    
    In the following, to differentiate the neighborhoods in $G$ and in $G_C$ or $G'$, we write them using either $N_G$, $N_{G_C}$, or $N_{G'}$ notations.
    Notice that, for each non-negative integer $t\le \max\set{k+\rho,r}$, we have $G[N_{G}^t[C]]\subseteq G_C$. This is because $C$ has weak diameter in $G$ at most $\ell = f(6k+6\rho) + 4k+4\rho$ and thus $G[N_{G}^t[C]]$ has weak diameter in $G$ at most $f(6k+6\rho) + 4k+4\rho + t \le T$, by the choice of $T$. Moreover, as $\sC$ is hereditary, we have that $G[N_{G}^t[C]]\in \sC$.

    As $G[N^t_G[C]]\subseteq G_C$, vertices in $C$ have the same distance $t$-neighborhood in $G$ and $G_C$.
    In particular, $G[N^t_G[C]] = G_C[N^t_{G_C}[C]]$ and $G[N^t_G[C_i]] = G'[N^t_{G'}[C_i]]$.
    From~\cref{fact:d_local_problem}, any minimum set of vertices of $G$ that satisfies the constraints of $N^k_G[C]$ is contained in $G[N^{k+\rho}_G[C]]$, and similarly if we replace $G$ by $G_C$.
    It follows that $\OPT_\Pi(G,N^k_{G}[C]) = \OPT_\Pi(G_C,N^k_{G_C}[C])$. 
    Recall that any two connected components $C, C' \in \CC_i$ are at distance in $G$ at least $2k+2\rho+1$. So $G[N^{k+\rho}_G[C]]$ and $G[N^{k+\rho}_G[C']]$ are disjoint subgraphs, and, as to satisfy all constraints of vertices in $V(G)$, one needs to satisfy constraints of vertices in $N^{k}_G[C]$ and $N^{k}_G[C']$, by disjunction of these sets and by~\cref{fact:d_local_problem}, we have by summing over all $C\in\CC_i$ that
    \[
        \OPT_\Pi(G,N^k_{G}[C_i]) = \OPT_\Pi(G',N^k_{G'}[C_i]).
    \]
    Moreover, because $\sC$ is stable by disjoint union and that $G[N_{G}^t[C]]\in\sC$ for every $C\in \CC'_i$, we get $G[N_{G}^t[C_i]]\in\sC$.
    As Algorithm~$\sfA$ is a $\frac{1}{d+1}$-weakly $k$-uniform $\alpha$-approximation on $\sC$, that $G'\in \sC$ and that $|C_i|\geq \big\lfloor\frac{|V(G)|}{d+1}\big\rfloor\geq \big\lfloor\frac{|V(G')|}{d+1}\big\rfloor$ (and therefore $G'$ indeed has polynomial identifiers), we get by definition of weak uniformity that
    \[
        \card{\sfA(G')\cap C_i} ~\le~ \alpha \cdot \OPT_\Pi(G',N_{G'}^k[C_i]) = \alpha \cdot \OPT_\Pi(G, N^k_G[C_i]) ~.
    \]
    Since $G[N^r_G[C_i]] = G'[N^r_{G'}[C_i]]$, Algorithm~$\sfA$, which runs in $r$ rounds, returns the same vertex set in $G$ and $G'$ for vertices in $C_i$. Therefore, $\card{\sfA(G)\cap C_i} = \card{\sfA(G')\cap C_i}$. Combining with the previous inequality, we get
    \begin{equation}
    \begin{aligned}\label{eq:nonpaddable_Ci_bound}
        \card{\sfA(G)\cap C_i} ~\le~ \alpha \cdot \OPT_\Pi(G, N_G^k[C_i])~\le~ &\alpha \cdot \OPT_\Pi(G,V(G))\\
        &~=~ \alpha \cdot \OPT_\Pi(G)~.
    \end{aligned}
    \end{equation}
    Putting \cref{eq:nonpaddable_Ci_bound} in \cref{eq:nonpaddable}, we get that Algorithm~$\sfA$ is an $\alpha(d+1)$-approximation on $\sD$.
\end{proof}

\subsection{Meta-theorem for paddable problems with non-empty error set}
\label{appendix:paddable_metatheorem}

\begin{restatable}[Meta-Theorem~III]{theorem}{MetaTheoremPaddable}\label{th:metatheorem_paddable}
    Let $d,k$ be integers $\alpha > 0$ a real, $\sfA$ be a $\eps$-weakly~\footnote{Note that, for the proof, we could also prove as a preliminary that weakly uniform and paddable implies uniform. However, for sake of brevity, we prove everything directly in the proof.} $k$-uniform $\alpha$-approximation \LOCAL algorithm for an additive $\rho$-local problem $\Pi$ in a class of graphs $\sC$ with round complexity $r$, where $\sC$ is hereditary, stable by disjoint union, and $O(1)$-paddable for $\Pi$.
    Let $\sD$ be a graph class of asymptotic dimension $d$ with $d$-dimensional control function $f$ such that $\Pi$ is dense in $\sD$, and such that $\sD$ is $T$-locally $\delta$-nice w.r.t. $\sC$ and $\Pi$, where $T = f(2k+2\rho) + \max\set{k+\rho,r}$. Then there exists a function $g$ such that for any $\eta>0$, there exists a $\max\{k,\rho\}$-uniform $(\alpha (d+1)+1+\eta)$-approximation algorithm $\sfB$ on $\sD$ with round complexity $\max\{T+\delta+2,g(\eta)\}$.
\end{restatable}

\begin{proof}
    Let $G\in \sD$ and $X$ be the set of $T$-errors of $G$ with respect to $\sC$.
    As $\sD$ has $d$-dimensional control function $f$, then $G^{2k+2\rho}$ admits a $(d+1)$-coloring $f(2k+2\rho)$-bounded in $G$.
    Fix such a coloring, and, for each $i\in\set{0,1,\dots,d}$, let $C_i$ be the set of color-$i$ vertices in $G^{2k+2\rho}$ (and also in $G$). 
    We denote by $\CC_i$ the set of $(2k+2\rho)$-components of $C_i$ (i.e. connected components of $G^{2k+2\rho}[C_i]$). 
    So, by definition, all components of $\CC_i$ have weak diameter in $G$ at most $f(2k+2\rho)$, and are pairwise at distance at least $2k+2\rho+1$ in $G$ (because distinct components of $\CC_i$ cannot be adjacent in $G^{2k+2\rho}$).

    We describe the Algorithm~$\sfB$ we want to use in \cref{algo:PiIII}, which is very similar to \cref{algo:MDS_bdd_genus}.
    \begin{algorithm}
    \caption{Generic algorithm $\sfB$ from \MetaIII.}\label{algo:PiIII}
    \begin{algorithmic}[1]

    \Require An $\eps$-weakly $k$-uniform $\alpha$-approximate \LOCAL algorithm $\sfA$ for $\Pi$ in $\sC$ with round complexity $r$, a graph classes $\sC$ and $\sD$ with properties of the statement of \MetaIII, the $d$-dimensional control function $f$ of $\sD$ and $G\in\sD$ with $T$-errors $X$ where $T = f(2k+2\rho) + \max\set{k+\rho,r}$.

    \Ensure A solution $S$ for the graph $G$ to the problem $\Pi$ with the guarantee that $\card{S}\leq(\alpha (d+1)+1) \cdot \OPT_\Pi(G)$.

    \State  Let $B$ be such that any graph $G$ of $\sD$ with $\diam(G)\geq B$ has $\OPT_\Pi(G)\geq \alpha(d+1)\omega/\eta$.
    If $\diam(G)<B$, brute-force an optimal solution. 
    \Comment{$B$ exists as $\Pi$ is dense in $\sD$.}

    \State $S \gets \emptyset$

    \State Each vertex $u$ computes $G[N^T[u]]$ and checks whether it belongs to $\sC$; as a consequence, it decides whether it belongs to $X$. 

    \State Each vertex $u$ runs $\sfA$ for $r$ rounds, and gets added to $S$ only if $u\in \sfA(G)\setminus X$.

    \State $S \gets S \cup S'$, where $S'$ is a brute-forced minimum set of $G$ that satisfies the constraints of vertices that are non-satisfied.

    \end{algorithmic}
    \end{algorithm}

    We run Algorithm~$\sfB$ on $G$. It is easy to check that $\sfB(G)$ is a valid solution for $G$, because $\Pi$ is additive. By definition, $\card{\sfB(G)} \le |S| + |S'|$, where $S = \sfA(G)\setminus X$ (Step~4) and $S'$ is a brute-forced set of minimum size satisfying the constraints of vertices of $G$ not satisfied (Step~5). Note that every vertex not in $N^{\rho}[X]$ is satisfied by $S$, so $S'\subseteq N^{2\rho}[X]$. Clearly, $\card{S'} \le \OPT_\Pi(G)$, and $\card{S'} = 0$ if there is no $T$-errors ($X=\emptyset$). Moreover, by definition of the coloring, $\card{\sfA(G)\setminus X} = \sum_{i=0}^d \card{(\sfA(G)\cap C_i)\setminus X}$. In other words,
    \begin{equation}\label{eq:B(G)paddable}
    \card{\sfB(G)} ~\le~ \sum_{i=0}^d \card{(\sfA(G)\cap C_i)\setminus X} + \left\{\begin{array}{ll}
        0 &\mbox{if $X=\emptyset$}\\
        \OPT_\Pi(G) &\mbox{otherwise}
        \end{array}
        \right.
    \end{equation}
    
    To upper bound the term $\card{(\sfA(G)\cap C_i)\setminus X}$, consider some color $i\in\set{0,1,\dots,d}$ and some component $C\in \CC_i$ such that $C\not\subseteq X$. Consider some $v_C\in C\setminus X$, and denote $G_C = G[N^T[v_C]]$. Note that, by definition of $X$ and as $v\notin X$, we have $G_C \in\sC$. Moreover, let $\CC'_i = \{C\in\CC_i \mid C\not\subseteq X\}$ and let $G' = G[\bigcup_{C\in\CC'_i} N^T[v_C]]$ be the union of the $G_C$'s for $C\in\CC_i$ such that $C\not\subseteq X$.

    Along the proof, we will use several times the following fact:

    \DLocalProblem*
    
    In the following, to differentiate the neighborhoods in $G$ and in $G_C$ or $G'$, we write them using either $N_G$, $N_{G_C}$, or $N_{G'}$ notations.
    Notice that, for each non-negative integer $t\le \max\set{k+\rho,r}$, we have $G[N_{G}^t[C]]\subseteq G_C$. This is because $C$ has weak diameter in $G$ at most $f(2k+2\rho)$ and thus $G[N_{G}^t[C]]$ has weak diameter in $G$ at most $f(2k+2\rho) + t \le T$, by the choice of $T$. Moreover, as $\sC$ is hereditary, we have that $G[N_{G}^t[C]]\in \sC$.

    As $G[N^t_G[C]]\subseteq G_C$, vertices in $C$ have the same distance $t$-neighborhood in $G$ and $G_C$.
    In particular, $G[N^t_G[C]] = G_C[N^t_{G_C}[C]]$ and $G[N^t_G[\bigcup\CC'_i]] = G'[N^t_{G'}[\bigcup\CC'_i]]$.
    From~\cref{fact:d_local_problem}, any minimum set of vertices of $G$ that satisfies the constraints of $N^k_G[C]$ is contained in $G[N^{k+\rho}_G[C]]$, and similarly if we replace $G$ by $G_C$.
    It follows that $\OPT_\Pi(G,N^k_{G}[C]) = \OPT_\Pi(G_C,N^k_{G_C}[C])$. 
    Recall that any two connected components $C, C' \in \CC_i$ are at distance in $G$ at least $2k+2\rho+1$. So $G[N^{k+\rho}_G[C]]$ and $G[N^{k+\rho}_G[C']]$ are disjoint subgraphs, and, as to satisfy all constraints of vertices in $V(G)$, one needs to satisfy constraints of vertices in $N^{k}_G[C]$ and $N^{k}_G[C']$, by disjunction of these sets and by~\cref{fact:d_local_problem}, we have by summing over all $C\in\CC_i$ that
    \[
        \OPT_\Pi\left(G,N^k_{G}\left[\bigcup \CC'_i\right]\right) = \OPT_\Pi\left(G',N^k_{G'}\left[\bigcup \CC'_i\right]\right).
    \]
    Moreover, because $\sC$ is stable by disjoint union and that $G[N_{G}^t[C]]\in\sC$ for every $C\in \CC'_i$, we get $G[N_{G}^t[\bigcup\CC'_i]]\in\sC$. 
    
    In the following, we will choose a graph $\widehat{G}$ with size linear in $|V(G)|$ so that adding $\widehat{G}$ to any subgraph $H$ of $G'$ will form a graph with size linear in  $|V(\widehat{G}\sqcup G)|$. 
    This serves two purposes: we are now able to apply weak uniformity on $\widehat{G}\sqcup H$ and we can attribute unique identifiers to $\widehat{G}$ so that the identifiers of $\widehat{G}\sqcup H$ are polynomial (we keep the same identifiers for $H$).
    Let $N:= \lceil\eps|V(G)|/(1-\eps)\rceil$. By paddability, there exists a graph $\widehat{G}\in\sC$ on at least $N$ vertices such that $\OPT_\Pi(\widehat{G})\leq \omega$.
    Let $I$ be the largest identifier of $G$. We give identifiers to vertices of $\widehat{G}$ from the set $\{I+1,I+2,\dots,I+|V(G')|\}$.
    Let $G''$ be the disjoint union of $G'$ and $\widehat{G}$ and let $Y=\bigcup\CC'_i \sqcup V(\widehat{G})$. The identifiers of vertices in $G''$ are indeed polynomial in the size of $G''$, as $\widehat{G}$ has size at least linear in $|V(G)|$. 
    Our goal now is to apply weak uniformity in $G''$ to the set $Y$. Note that we only care about $\card{\sfA(G')\cap \bigcup\CC'_i}$ but need $Y$ to apply weak uniformity.
    We first prove that $|Y|\geq \eps|V(G'')|$.
    \begin{align*}
        |Y|&\geq |V(\widehat{G})| \geq \eps|V(\widehat{G})| + (1-\eps)|V(\widehat{G})|\\&\geq \eps|V(\widehat{G})|+(1-\eps)N\geq\eps|V(\widehat{G})|+\eps|V(G)|\geq \eps|V(G'')|~.
    \end{align*}
    Therefore, by $\eps$-weak $k$-uniformity $\alpha$-approximation of $\sfA$ on $\sC$, we get
    \begin{align*}
        \card{\sfA(G'')\cap Y}\leq \alpha\cdot\OPT_\Pi(G'',N^k[Y])~.
    \end{align*}
    We have $|\sfA(G'')\cap Y| = |\sfA(G')\cap \bigcup\CC'_i| + |\sfA(\widehat{G})|$ and $\OPT_\Pi(G'',N_ {G''}^k[Y])=\OPT_\Pi(\widehat{G}) + \OPT_\Pi(G',N_{G'}^k[\bigcup\CC'_i])\leq \omega+\OPT_\Pi(G',N_{G'}^k[\bigcup\CC'_i])$ because $\widehat{G}$ and $G'$ are in disjoint components of $G''$. Forgetting about the term $|\sfA(\widehat{G})|$, it follows that 
    \begin{align}\label{eq:omega}
            \card{\sfA(G')\cap \bigcup\CC'_i}\leq \alpha\cdot \omega +\alpha\cdot\OPT_\Pi\left(G',N_{G'}^k\left[\bigcup\CC'_i\right]\right)~.
    \end{align}

    Since $G[N^r_G[\bigcup\CC'_i]] = G'[N^r_{G'}[\bigcup\CC'_i]]$, Algorithm~$\sfA$, which runs in $r$ rounds, returns the same vertex set in $G$ and $G'$ for vertices in $\bigcup\CC'_i$. Therefore, $\card{\sfA(G)\cap \bigcup\CC'_i} = \card{\sfA(G')\cap \bigcup\CC'_i}$. Combining with the previous inequality, we get
    \def\ccix{\substack{C\in\CC_i\\ C\not\subseteq X}}

    \begin{align*}
        \card{\sfA(G)\cap \bigcup\CC'_i} ~\le~ \alpha\cdot\omega+\alpha \cdot \OPT_\Pi\left(G, N_G^k\left[\bigcup\CC'_i\right]\right)~\le~ &\alpha\cdot\omega+\alpha \cdot \OPT_\Pi(G,V(G))\\
        &~=~ \alpha\cdot\omega+\alpha \cdot \OPT_\Pi(G)~.
    \end{align*}
    We remark that $\card{(\sfA(G)\cap C)\setminus X} = 0$ if $C \subseteq X$. Because all vertices in $C_i$ but not in $X$ are in some $C\in\CC'_i$, we have
    
    \begin{align}\label{eq:paddable_Ci_bound}
        \card{(\sfA(G)\cap C_i)\setminus X} \leq \card{\sfA(G)\cap \bigcup\CC'_i} \leq \alpha\cdot\omega+ \alpha \cdot \OPT_\Pi(G)~.
    \end{align}

    Now, we can apply the same trick as in \cref{prop:aMDS+b}.
    As $\Pi$ is dense in $\sD$, there exists some $B$ such that any graph of $\sD$ of diameter at least $B$ has optimum value at least $\alpha(d+1)\omega/\eta$.
    Algorithm~$\sfB$ checks (Step~1) whether the diameter $D$ of its component is less than $B$. This can be done by collecting the radius-$B$ neighborhood of every vertex in $B$ rounds. If true, the vertex locally brute-forces an optimal solution for its component, which it already knows about. If false, the vertex continues the algorithm with Step~2.
    If $D<B$, then the algorithm is obviously a good approximation (with approximation ratio $1$).
    Otherwise, $D\geq B$ and it follows that $\OPT_\Pi(G)\geq \alpha(d+1)\omega/\eta$. Putting in \cref{eq:paddable_Ci_bound} in \cref{eq:B(G)paddable}, we get that Algorithm~$\sfB$ is an $(\alpha(d+1)+1+\eta)$-approximation, and even an $(\alpha(d+1)+\eta)$-approximation if $X=\emptyset$. 
    
    To see that $\sfB$ is also a $k$-uniform approximation, it is sufficient to see that, for any $W \subseteq V(G)$, the bound obtained is $\card{\sfB(G)\cap W} \le \alpha(d+1)\omega +\alpha (d+1)\cdot \OPT_\Pi(G, N_G^k[W]) + \OPT_\Pi(G, N_G[W])$ where the first two terms are due to running $\sfA$ and the last term is due to running the brute-force.
    Indeed, if the brute-force computed a set whose intersection with $W$ was smaller than $\OPT_\Pi(G, N^\rho_G[W])$, we could replace it by the minimum size set satisfying the constraints of the vertices in $N^\rho_G[W]$ and obtain a smaller set, a contradiction. This proves that $\sfB$ is $\max\{k,\rho\}$-uniform with approximate ratio $\alpha(d+1)+1$ (or $\alpha(d+1)$ if $X=\emptyset$, as the brute-force is not required in that case).
    Now, let us prove that Algorithm~$\sfB$ has the desired round complexity.
    \begin{itemize}
        \item Step~1 takes $B = g(\eta)$ rounds for computing if the diameter is small, and no additional round for the brute-force.

        \item Computing $X$ in Step~3 takes $T+1$ rounds.
    
        \item Running $\sfA$ in Step~4 takes~$r$ rounds.
    
        \item For Step~5, consider the set $S$ computed by $\sfB$ at Step~4, before the brute-force.
        Observe that the vertices in the set $W = V(G)\setminus N^\rho_G[X]$ are satisfied by $S$. This is because vertices of $W$ are at distance at least $\rho+1$ from $X$ and thus vertices of $N^\rho_G[W]$ cannot be in $X$. Thus, $\sfA$ applies to all vertices of $N^\rho_G[W]$, and so $W$ is indeed satisfied by $S$ (\cref{fact:d_local_problem}).
        Therefore, to satisfy $N^\rho_G[X]$ it is enough to select a set $S'$ from $N^{2\rho}_G[X]$ (\cref{fact:d_local_problem}) as done in Step~5.
        By assumption, the connected components of $N_G^{2\rho}[X]$ have weak diameter in $G$ at most $\delta$ (obviously, if $X = \emptyset$ then $\delta = 0$). Therefore, the brute-force (Step~5) will take at most $\delta+1$ rounds. 
    \end{itemize}
    Naively, Steps~3 and~4 together take $(T+1) + r$ rounds. As Algorithm~$\sfA$ is not guaranteed to work when executed on vertices of $X$, Step~4 must be run only after Step~3. However, both steps can be run in parallel as follows. We run $\sfA$ for $r$ rounds exactly and stop it just after (since its running time on a vertex of $X$ could result into more than $r$ rounds). Then, the decision to add $u$ in $S$ (if selected by $\sfA$) is delayed for the next $T+1 - r$ rounds (this is non-negative from the choice of $T$). In parallel, $G[N^T[u]]$ is computed and is checked to be in $\sC$ or not after $T+1$ rounds. So, after $\max\set{r,T+1} = T+1$ rounds, set $S$ in Step~4 has been completed.
    The same trick can be applied for Step~1.
    Step~5 takes $\delta+1$ steps, so that the total round complexity of $\sfB$ is $\max\{T+\delta+2,g(\eta)\}$, completing the proof.
\end{proof}

\section{Miscellaneous propositions}

\begin{proposition}\label{prop:0uniform}
    Every graph class $\sC$ including trees has no $0$-uniform $\alpha$-approximations for \minDS, for every ratio $\alpha \ge 1$.
\end{proposition}

\begin{proof}
Consider a depth-$2$ tree $T_{\alpha}$, with $\alpha + 1$ vertices at depth-$1$, each having $\alpha^2 + 1$ neighbors at depth-2. See \cref{fig:0_uniform_tree} for an example. Clearly, $\MDS(T_{\alpha}) = \alpha+1$. Moreover, any $\alpha$-approximation algorithm $\sfA$ in $T_{\alpha}$ has to select all depth-$1$ vertices. 
If not, each of the $t$ non-selected vertex at depth~$1$ would force the selection of $\alpha^2+1$ extra vertices at depth~$2$, creating by this way a solution with at least $(\alpha+1-t) + t\cdot (\alpha^2+1) = \alpha\cdot(t\alpha+1) + 1$ vertices. For $t\ge 1$, this is at least $\alpha\cdot(\alpha+1)+1 > \alpha \cdot \MDS(T_{\alpha})$. 
However, for a subset $S$ composed of all depth-$1$ vertices of $T_{\alpha}$ (so with $\card{S} = \alpha+1$), we have on one side $\card{\sfA(T_{\alpha}) \cap S} = \card{S} = \alpha+1$, whereas $\MDS(T_{\alpha},N^0[S]) = \MDS(T_{\alpha},S) = 1$ (considering the root). So $\card{\sfA(T_{\alpha}) \cap S} \not\le \alpha \cdot \MDS(T_{\alpha},N^0[S])$. 
Therefore, $\sfA$ cannot be $0$-uniform in $T_{\alpha}$.
\begin{figure}[htbp!]
    \centering \begin{tikzpicture}[scale=0.7]
  \tikzset{every node/.style={vertex}}
  \tikzset{every path/.style={edge}}
\node (0) at (0,0) {}; \def\N{5} \def\LF{5.7} \draw[fill=mygreen,rounded corners,draw=none,opacity=0.35] ({-\LF-1},-1) rectangle ({\LF+1},-3);
  \foreach \F in {1,2,3} { \def\XF{\LF*(\F-2)} \node[draw=mygreen] (\F) at ({\XF},-2) {};
    \draw (0) -- (\F);
    \foreach \i in {1,...,\N} {
      \def\XI{0.85*(\i-(\N+1)/2)} \node (\F\i) at ({\XF+\XI},-5) {};
      \draw (\F\i) -- (\F);
    }
    }
    \node[draw=none,fill=none] at ({\LF+1.5},-0.75) { \color{mygreen!70!Green}{\LARGE $S$}};
\end{tikzpicture}
 \caption{The tree $T_{\alpha}$ has no $0$-uniform $\alpha$-approximations (here with $\alpha=2$). Indeed, any $\alpha$-approximation $\sfA$ of its minimum dominating set must contain the $\alpha+1$ vertices of $S$ (green), and thus $\card{\sfA(T_{\alpha}) \cap S} = \card{S} \not\le \alpha \cdot \MDS(T_{\alpha},N^0[S])$ since $S$ can be dominated by a single vertex (the root).}\Description{~}
 \label{fig:0_uniform_tree}
 \end{figure}
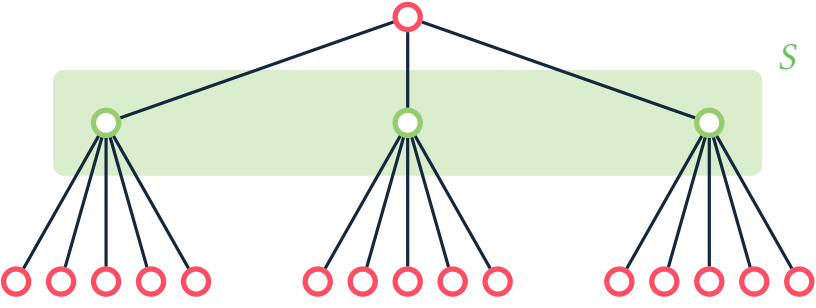
\end{proof}

\begin{proposition}\label{prop:aMDS+b}
    For every $\eps>0$, any $r$-round algorithm that returns a dominating set of size at most $\alpha \cdot \MDS(G) + \beta$ on $G$, can be converted into an $(\alpha+\eps)$-approximation \LOCAL algorithm with round complexity $r + O(\beta/\eps)$.
\end{proposition}

\begin{proof}
    Observe that the diameter of each connected component of the graph is at most $3\cdot\MDS(G)$. So, each vertex can check whether the diameter $D$ of its component is less than $B$, for $B = 3\beta/\eps$. This can be done by collecting its radius-$B$ neighborhood in $B = O(\beta/\eps)$ rounds. If true, the vertex locally brute forces an optimal solution for its component that, which it already knows about. If false, the vertex applies the approximate algorithm. In that case, $B\le D$ and $D \le 3 \MDS(G)$, which together imply $\beta \le \eps \cdot\MDS(G)$. Thus, the returned set has size at most $\alpha\cdot\MDS(G) + \beta \le (\alpha+\eps)\cdot\MDS(G)$.    
\end{proof}

\begin{proposition}\label{prop:nabla1}
    Every $n$-vertex graph $G$ of Euler genus $g$ satisfies $\nabla_1(G) < \sqrt{3g/2} + 3$. This bound is best possible up to an additive constant, for each $g\ge 0$.
\end{proposition}

\begin{proof}
    By definition, we have $\nabla_1(G) = \max_{H} \set{ \card{E(H)}/\card{V(H)} }$, where $H$ is a depth-1 minor. For every $H$, we have $\card{E(H)}/\card{V(H)} < 3g/\card{V(H)} + 3$ from Euler's formula. As $\card{E(H)} \le \binom{\card{V(H)}}{2}$, $\nabla_1(G) \le \max_{m\in [1,n]} \min\set{(m-1)/2, 3g/m + 3} < \sqrt{3g/2} + 3$.

    The latter inequality holds, because either $3g/m + 3 < \sqrt{3g/2} + 3$, and we are done, or $3g/m + 3 \ge \sqrt{3g/2} + 3$, which implies $m \le \sqrt{6g}$. In that case $(m-1)/2 \le (\sqrt{6g}\,)/2 - 1/2 < \sqrt{3g/2} + 3$ as wanted.

    Consider a graph $G$ composed of a clique $K_m$ with $m\ge 4$ and $m\neq 7$ where each edge is replaced by a path of length three. The graph $G$ has same Euler genus as $K_m$, which is~$g = \ceil{(m-3)(m-4)/6}$ (cf.~\cite[Theorem~4.4.5]{MT01}). Moreover, it has $K_m$ as depth-$1$ minor. Thus, $\nabla_1(G) \ge \card{E(K_m)}/m = (m-1)/2 \ge  \sqrt{3g/2 + 3/4}$.
    
    The latter inequality holds because, one can check that $(m-1)^2 \ge (m-3)(m-4) + 9$ for all $m\ge 4$. Since $(m-3)(m-4) + 9 \ge 6\ceil{(m-3)(m-4)/6} + 3 = 6g+3$, we have $(m-1)^2/4 \ge (6g+3)/4$, and thus $(m-1)/2 \ge  \sqrt{3g/2 + 3/4}$ as claimed.
\end{proof}

\end{document}